\documentclass[12pt, draftclsnofoot, onecolumn]{IEEEtran}

\usepackage{setspace}
\usepackage{amsmath}
\usepackage{amsfonts}
\usepackage{graphicx}
\usepackage{amsthm}
\usepackage{amssymb}
\usepackage{mathrsfs}
\usepackage{stfloats}
\usepackage{geometry}
\usepackage{color}
\usepackage{bm}
\usepackage{subfigure}

\newtheorem{lemma}{Lemma}
\newtheorem{theorem}{Theorem}

\newtheorem{definition}{Definition}
\newtheorem{proposition}{Proposition}
\newtheorem{corollary}{Corollary}

\begin{document}

\title{Secure Communication for Spatially Sparse Millimeter-Wave Massive MIMO Channels via Hybrid Precoding}


\author{\small
\IEEEauthorblockN{
Jindan Xu$^1$,~\emph{Student Member,~IEEE},
Wei Xu$^1$,~\emph{Senior Member,~IEEE},\\
Derrick Wing Kwan Ng$^2$,~\emph{Senior Member,~IEEE},
and A. Lee Swindlehurst$^3$,~\emph{Fellow,~IEEE}\\
\IEEEauthorblockA{
$^1$National Mobile Communications Research Laboratory, Southeast University, Nanjing 210096, China\\
$^2$School of Electrical Engineering and Telecommunications, University of New South Wales, NSW 2052, Australia\\
$^3$Center for Pervasive Communications and Computing, University of California, Irvine, CA 92697, USA\\
Email: \{jdxu, wxu\}@seu.edu.cn, w.k.ng@unsw.edu.au, swindle@uci.edu}
}}

\maketitle

\begin{abstract}
In this paper, we investigate secure communication over sparse millimeter-wave (mm-Wave) massive multiple-input multiple-output (MIMO) channels by exploiting the spatial sparsity of legitimate user's channel. We propose a secure communication scheme in which information data is precoded onto dominant angle components of the sparse channel through a limited number of radio-frequency (RF) chains, while artificial noise (AN) is broadcast over the remaining nondominant angles interfering only with the eavesdropper with a high probability. It is shown that the channel sparsity plays a fundamental role analogous to secret keys in achieving secure communication. Hence, by defining two statistical measures of the channel sparsity, we analytically characterize its impact on secrecy rate. In particular, a substantial improvement on secrecy rate can be obtained by the proposed scheme due to the uncertainty, i.e., ``entropy'', introduced by the channel sparsity which is unknown to the eavesdropper. It is revealed that sparsity in the power domain can always contribute to the secrecy rate. In contrast, in the angle domain, there exists an optimal level of sparsity that maximizes the secrecy rate. The effectiveness of the proposed scheme and derived results are verified by numerical simulations.

\end{abstract}


\IEEEpeerreviewmaketitle
\section{Introduction}

Secure communication has recently attracted pervasive interest for safeguarding multifarious wireless services in cellular networks.
Traditionally, communication security is realized by cryptographic encryption algorithms implemented at network and application layers \cite{Introduction}.
These techniques are vulnerable since they rely on the assumption that adversaries have limited computational ability.
On the other hand, it was discovered in \cite{Shannon} that in theory, perfect secrecy can be achieved at the physical layer, which serves as a complement technology to conventional security methods.


In the past decade, physical layer security has attracted significant interests \cite{Survey_PLS}-\cite{Sec_QianXu}.
In \cite{Wiretap}, the well-known wiretap channel model was studied by Wyner.
In this model, the channel between the transmitter and eavesdropper is assumed to be a degraded version of the channel between the transmitter and legitimate user, guaranteeing a nonnegative secrecy capacity.
Later in \cite{Sec_broadcast}, the authors investigated a general case where the two channels are independent. It was revealed that secure communication is achievable if the capacity of the legitimate user's channel is larger than that of the eavesdropper's channel.
As a result, when the eavesdropper's channel happens to be stronger than that of the legitimate user's channel, artificial noise (AN) has to be designed and exploited properly to guarantee secure transmission \cite{Sec_MIMO_AN}-\cite{Sec_QianXu}.
More specifically, the AN in \cite{Sec_MIMO_AN} was designed to be transmitted in the space orthogonal to the legitimate user's channel such that only the eavesdropper's channel is impaired.
The authors of \cite{Sec_MIMO_AN_3} then extended this orthogonal AN design to scenarios with multiple eavesdroppers.
In addition, a joint design of AN and information carrying signals was studied in \cite{Sec_QianXu}.

With the recent development of multiple-input multiple-output (MIMO) techniques, physical layer security in the context of massive MIMO systems has been intensively investigated, e.g., \cite{hmWang2019Secure}-\cite{lsFan2017Secure}.
The authors of \cite{ypWu2016Secure} investigated secure transmission strategies when an active eavesdropper is present in a multiuser massive MIMO system.
In \cite{Secure_Zhu_1}, secure communication using maximal-ratio-transmission (MRT) precoding was studied for a multi-cell massive MIMO network.
Then, in \cite{Secure_Zhu_2}, the authors further studied the secrecy performance adopting more sophisticated precoding schemes, i.e., zero-forcing (ZF), regularized channel inversion (RCI), and collaborative ZF/RCI precoders. 
The impacts of low-resolution digital-to-analog converters (DACs) and spatial channel correlation on secure massive MIMO communications were respectively studied in \cite{J_Xu_Secure} and \cite{lsFan2017Secure}.

Along with massive MIMO, millimeter-wave (mm-Wave) communication also has been recognized as a promising solution to realizing ultra-high data rates for next-generation wireless networks, e.g., \cite{mmWave_Beamforming}-\cite{Kwan2017_5G}. Its applications to guarantee communication security have attracted significant attentions \cite{Secure_mmWave_antenna}-\cite{YongxuZhu2017Secure}.
Particularly in \cite{Secure_mmWave_antenna}, a low-complexity directional modulation technique was developed for point-to-point secure mm-Wave communications.
In \cite{Secure_mmWave_SPA}, a wireless transmission architecture, referred to as a switched phased-array (SPA), was proposed for mm-Wave systems to enhance physical layer security.
The authors of \cite{Secure_mmWave_Hybrid} analyzed the outage probability of secure communication in a mm-Wave overlaid microwave network in the presence of blockages.
The performance in terms of network-wide physical layer security was studied in \cite{ChaoWang2016PLS} for downlink transmission in a mm-Wave cellular network.
Moreover in \cite{YongxuZhu2017Secure}, the impacts of AN, blockages, and antenna gains on the system secrecy performance were characterized for a mm-Wave ad-hoc network.
Unlike conventionally adopted sub-6 GHz frequency bands in current cellular applications, mm-Wave channels are in general dominated by line-of-sight (LoS) components \cite{mmWave_propagation1}, \cite{mmWave_propagation2}.
In other words, spatial sparsity commonly exists in mm-Wave channels leading to new challenges in designing efficient communication systems \cite{Sparsity1}-\cite{Sparsity3}.
However, from the perspective of physical layer security, the sparsity of mm-Wave channels can be exploited to provide potential benefits.
Specifically, the sparsity of the legitimate user's channel depends on certain parameters, such as the angles of the dominant propagation directions which are unknown to eavesdroppers.
This channel-specific sparsity information can be leveraged as an implicit secret key which reduces the leakage of private information to potential eavesdroppers.

To the best of our knowledge, most existing works (e.g., \cite{hmWang2019Secure}-\cite{lsFan2017Secure}, \cite{Secure_mmWave_antenna}-\cite{YongxuZhu2017Secure}) have studied secrecy performance for massive MIMO communications from the viewpoint of time and frequency domains. However, the impact of channel sparsity has rarely been investigated.
In this paper, we investigate secure massive MIMO communications over spatially sparse mm-Wave channels and propose a secure transmission scheme implemented in angle domain.
By exploiting the sparsity pattern of the legitimate user's channel, confidential signals are transmitted over the dominant angular directions via a limited number of radio-frequency (RF) chains.
Secure communication is guaranteed by injecting AN into the remaining nondominant propagation angles.
In practice, the dominant propagation angles of the legitimate user's channel are different from that of the eavesdropper, and hence, the AN would cause less interference at the desired user but significant interference at the eavesdropper.
Based on our derived expressions of the secrecy rate, we quantitatively characterize the impact of channel sparsity on the system secrecy rate.
The main contributions of our work are summarized as follows.

1)
We show that the secrecy rate of mm-Wave massive MIMO systems can benefit from channel sparsity. This benefit is quantitatively characterized in closed-form as an additive secrecy rate bonus which can be interpreted as \emph{sparsity information}.
In particular, we propose two metrics, $\chi_\mathrm{L}$ and $\chi_\mathrm{H}$, for analyzing the impact of channel sparsity on the secrecy rate.
For low signal-to-noise ratios (SNRs), the secrecy rate bonus is $M_t\log_2 \chi_\mathrm{L}$ bits/s/Hz, where $M_t$ is the number of RF chains employed at the transmitter,
$\chi_\mathrm{L}(\rho,\eta)\triangleq \eta^{\rho-1}\left[\eta+(1-\eta)\rho\right]$ is a measure of channel sparsity while $\rho$ and $\eta$ represent, respectively, the sparsity in the angle domain and sparsity in the power domain.
For high SNRs, the secrecy rate bonus equals $M_t\log_2 \chi_\mathrm{H}$ bits/s/Hz where $\chi_\mathrm{H}(\rho,\eta)\triangleq \eta^{\rho-2}\left[\eta+(1-\eta)\rho\right]\left(\!1\!-\!\frac{M_r}{N_t(1-\rho)}\!\right)^{-1}$, $N_t$ is the number of antennas at the transmitter, and $M_r$ is the number of RF chains at the legitimate user.

2)
We analyze the effect of channel sparsity on secrecy rate from two perspectives, i.e., the values of $\rho$ and $\eta$ in the angle and power domains, respectively.
The smaller the values of $\rho$ and $\eta$ are, the sparser the channel is.
From the derived expressions of $\chi_\mathrm{L}$ and $\chi_\mathrm{H}$, we observe that the secrecy rate first increases and then decreases with $\rho$.
The maximum secrecy rate is achieved by a specific sparsity value in the angle domain given by $\rho^*=-\frac{1}{\ln\eta}-\frac{\eta}{1-\eta}$ for a fixed $\eta$.
On the other hand, the secrecy rate monotonically decreases with $\eta$, which implies that channel sparsity in the power domain always benefits the secrecy rate of the system.


The rest of the paper is structured as follows.
The channel model is introduced and the secure communication scheme is presented in Section~\uppercase\expandafter{\romannumeral2}.
In Section~\uppercase\expandafter{\romannumeral3}, we derive the ergodic achievable secrecy rate in the context of massive MIMO.
Two upper bounds for the secrecy rate are obtained in Section~\uppercase\expandafter{\romannumeral4}, based on which we quantitatively characterize the effect of channel sparsity.
Simulation results are presented in Section \uppercase\expandafter{\romannumeral5} and conclusions are drawn in Section~\uppercase\expandafter{\romannumeral6}.

\emph{Notation}:
$\mathbf{A}^T$, $\mathbf{A}^*$, and $\mathbf{A}^H$ represent the transpose, conjugate, and conjugate transpose of $\mathbf{A}$, respectively.
$\mathbf{a}\sim \mathcal{CN}(\mathbf{0},\mathbf{\Sigma})$ denotes a circularly symmetric complex Gaussian vector with zero mean and covariance $\mathbf{\Sigma}$.
$[\mathbf{A}(i,j)]_{ (i,j) \in \mathcal{S}}$ represents a matrix consisting of elements in $\mathbf{A}$ with indices belonging to $\mathcal{S}$ while $[\mathbf{a}(i)]_{ i \in \mathcal{S}}$ represents a vector consisting of elements in $\mathbf{a}$ with indices belonging to $\mathcal{S}$.
$\textrm{Tr}(\mathbf{A})$ and $|\mathbf{A}|$ are respectively the trace and determinant of $\mathbf{A}$.
$\mathbb{E}\{\cdot\}$ is the expectation operator.
$\xrightarrow{\mathrm{a.s.}}$ denotes almost sure convergence.
$[x]^+={\mathrm {max}}\{0,x\}$ returns the maximum of $0$ and $x$.

\section{System Model}

\begin{figure*}[tb]
\centering\includegraphics[width=0.95\textwidth]{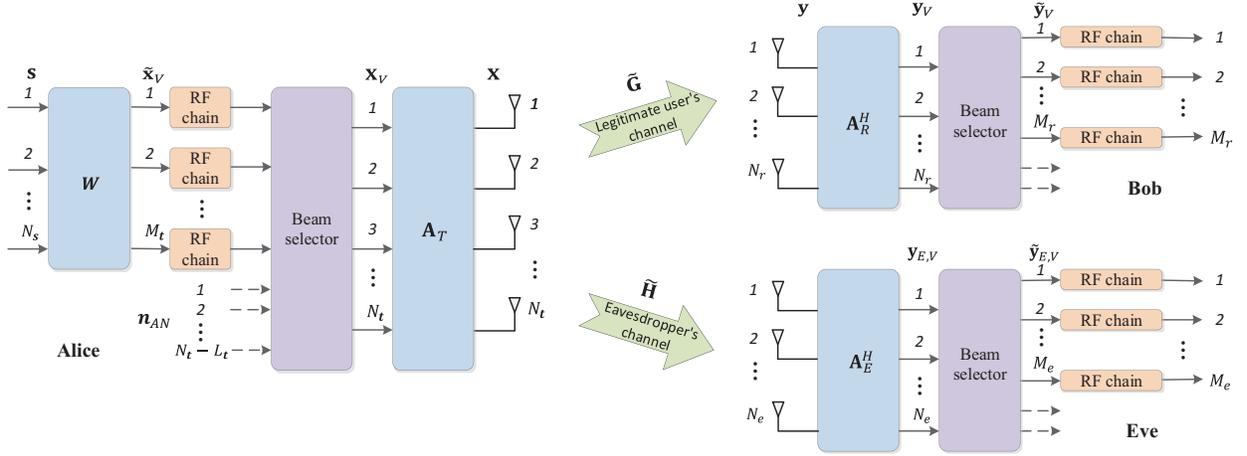}
\caption{System block diagram with a transmitter (Alice) communicating with a legitimate receiver (Bob) in the existence of a potential eavesdropper (Eve).}
\label{Block}
\end{figure*}

We consider a mm-Wave massive MIMO network, as depicted in Fig. \ref{Block}, where the transmitter (Alice), legitimate user (Bob), and passive eavesdropper (Eve) possess $N_t$, $N_r$, and $N_e$ antennas, respectively.
The devices are respectively equipped with $M_t$, $M_r$, and $M_e$ RF chains, where we generally have $M_t\leq N_t$, $M_r\leq N_r$, and $M_e\leq N_e$ in typical massive MIMO communication systems~\cite{Massive_Liang}.
Note that a hybrid architecture is implemented at Alice, including a digital precoder $\mathbf{W}$ and an analog precoder $\mathbf{A}_T$.


\subsection{Channel Model}

In order to express the channel using the angle domain decomposition \cite{channel}, we define the array response matrices at Alice and Bob respectively as $\mathbf{A}_T=\!\frac{1}{\sqrt{N_t}}\left[\mathbf{a}_T\left(\bar{\theta}_{T,1}\right),...,\mathbf{a}_T\left(\bar{\theta}_{T,N_t}\right)\right]$ and
$\mathbf{A}_R\!\!=\!\frac{1}{\sqrt{N_r}}\left[\mathbf{a}_R\left(\bar{\theta}_{R,1}\right),...,\mathbf{a}_R\left(\bar{\theta}_{R,N_r}\right)\right]$, where $\mathbf{a}_T(\bar{\theta}_{T,j})=  [1,e^{-j2\pi\bar{\vartheta}_{T,j}},...,e^{-j2\pi(N_t-1)\bar{\vartheta}_{T,j}}]^T$ and
$\mathbf{a}_R(\bar{\theta}_{R,i})= [\!1,e^{-j2\pi\bar{\vartheta}_{R,i}},...,e^{-j2\pi(N_r-1)\bar{\vartheta}_{R,i}}\!]^T$. Variables $\bar{\theta}_{T,j}\!=\!\arcsin\!\left(\!\frac{\lambda \bar{\vartheta}_{T,j}}{d}\!\right)$ and $\bar{\theta}_{R,i}\!=\!\arcsin\!\left(\!\frac{\lambda \bar{\vartheta}_{R,i}}{d}\!\right)$ are respectively the angle-of-departure (AOD) and angle-of-arrival (AOA) corresponding to uniformly spaced spatial angles $\bar{\vartheta}_{T,j}=\frac{j-1-(N_t-1)/2}{N_t}$ and $\bar{\vartheta}_{R,i}=\frac{i-1-(N_r-1)/2}{N_r}$.
Parameter $d$ is the distance between adjacent antennas and $\lambda$ is the wavelength of the information carrier frequency.
Then, the flat fading channel from Alice to Bob can be expressed as \cite{channel}
\begin{align}
\tilde{\mathbf{G}}=\frac{1}{\sqrt{N_r N_t}}\sum\limits_{i=1}^{N_r} \sum\limits_{j=1}^{N_t} \mathbf{G}(i,j) \mathbf{a}_R(\bar{\theta}_{R,i}) \mathbf{a}_T^H(\bar{\theta}_{T,j})
=&\mathbf{A}_R \mathbf{G} \mathbf{A}_T^H
,
\label{G}
\end{align}
where $\mathbf{G}\in \mathbb{C}^{N_r\times N_t}$ is usually called the equivalent virtual channel representation \cite{Sayeed2002Deconstructing}.
Since $\mathbf{A}_R$ and $\mathbf{A}_T$ are unitary discrete Fourier transform (DFT) matrices, we have $\mathbf{G}= \mathbf{A}_R^H \tilde{\mathbf{G}} \mathbf{A}_T$.

For the mm-Wave massive MIMO channel, sparsity implies that there are only a few significant non-zero coefficients in the virtual channel matrix $\mathbf{G}$, which we refer to as dominant beams~\cite{Recon0}.
According to \cite{Recon}, a low-dimensional virtual representation of the channel is available if the channel sparsity is present.
Let $\mathcal{U}$ denote the selection mask which contains all the indices of the dominant beams in $\mathbf{G}$.
The low-dimensional virtual representation of the channel is defined as
\begin{align}
\mathbf{G}^l=\left[ \mathbf{G}(i,j)\right]_{(i,j)\in \mathcal{U}}\in \mathbb{C}^{L_r\times L_t},
\label{Hv_til}
\end{align}
where $L_r$ and $L_t$ represent the numbers of dominant beams from the view of Bob and Alice, respectively.
Assuming that we exploit fewer RF chains than the number of dominant beams\footnote{If there are more RF chains than dominant beams, a cost-effective way is to exploit the same number of RF chains as dominant beams, i.e., $M_t = L_t$ and $M_r= L_r$.}, we have $M_t\leq L_t\leq N_t$ and $M_r\leq L_r\leq N_r$.


Similarly, the virtual channel matrix for Eve is given by
\begin{align}
\mathbf{H}
=\mathbf{A}_{E}^H \tilde{\mathbf{H}} \mathbf{A}_T
,
\label{Hv_e}
\end{align}
where $\tilde{\mathbf{H}}\in \mathbb{C}^{N_e\times N_t}$ is the flat fading channel between Alice and Eve,
and $\mathbf{A}_{E}=\frac{1}{\sqrt{N_e}}\left[\mathbf{a}_E\left(\bar{\theta}_{E,1}\right),...,\right.$ $\left.\mathbf{a}_E\left(\bar{\theta}_{E,N_e}\right)\right]$ is the array response matrix at Eve with $\mathbf{a}_E(\bar{\theta}_{E,i})\!= \! [\!1,e^{-j2\pi\bar{\vartheta}_{E,i}},...,e^{-j2\pi(N_e-1)\bar{\vartheta}_{E,i}}\!]^T$, $\bar{\theta}_{E,i}=\arcsin\left(\frac{\lambda \bar{\vartheta}_{E,i}}{d}\right)$, and $\bar{\vartheta}_{E,i}=\frac{i-1-(N_e-1)/2}{N_e}$.
Then, the corresponding low-dimensional virtual channel is
\begin{align}
\mathbf{H}^l=\left[ \mathbf{H}(i,j)\right]_{(i,j)\in \mathcal{E}}\in \mathbb{C}^{L_e\times L_t},
\label{Hev_til}
\end{align}
where $\mathcal{E}$, analogous to $\mathcal{U}$ for $\mathbf{G}$, denotes the selection mask of $\mathbf{H}$ and $L_e$ is the number of dominant receive beams at Eve.
 Similar to $L_t$ and $L_r$, we have $M_e\leq L_e\leq N_e$.



\subsection{Secure Transmission over a Sparse Channel}

Due to the limited number of RF chains exploited at Alice, Bob, and Eve, it is necessary to first choose a subset of the beam directions before transmission over these RF chains.
Similar to \cite{Recon0} and \cite{Recon}, the statistical channel state information (CSI) for $\mathbf{G}$, in the term of the sparsity pattern $\mathcal{U}$, is assumed to be known to Alice and Bob for beam selection before data transmission.
Also, Eve is assumed to know only the statistical CSI for $\mathbf{H}$, in the form of the sparsity pattern $\mathcal{E}$, for beam selection.
Estimation of this type of statistical channel information has been studied in a number of papers such as \cite{Statistical_CE_1}-\cite{Statistical_CE_3} for mm-Wave channels.
Given the independence of the channels from Alice to Bob and Eve \cite{Secure_Zhu_1}-\cite{J_Xu_Secure}, the sparsity patterns, $\mathcal{U}$ and $\mathcal{E}$, are assumed to be distinct.
By exploiting the mismatch between $\mathcal{U}$ and $\mathcal{E}$, we propose a secure communication scheme in which confidential signals are aligned with the dominant beams of $\mathbf{G}$ while AN is sent over its nondominant beams.
From the perspective of Eve, however, the confidential signals and AN are randomly transmitted through the dominant and nondominant beams of $\mathbf{H}$.
In this way, the channel of Eve is deliberately degraded and secure communication can be achieved.
For detection, however, it is generally possible to acquire fairly accurate instantaneous CSI of the selected sub-channels through channel estimation \cite{Sec_MIMO_AN}, \cite{Secure_Zhu_1}.
The instantaneous CSI can be estimated by exploiting channel training via pilots.

A description of this scheme is depicted in Fig. \ref{Block}.
Let $\mathbf{s} \in \mathbb{C}^{N_s\times 1}$ denote the normalized confidential signals for Bob. 
The transmitter first generates $\widetilde{\mathbf{x}}_V=\mathbf{Ws}\in \mathbb{C}^{M_t\times 1}$ by using precoding matrix $\mathbf{W}\in \mathbb{C}^{M_t\times N_s}$.
Then, the precoded signal $\widetilde{\mathbf{x}}_V$ goes through a beam selector to form a virtual signal vector, $\mathbf{x}_V \in \mathbb{C}^{N_t\times 1}$, where AN is also included to combat the channel quality of Eve.
Given $\mathcal{U}=\{(i,j)|i\in\mathcal{U}_r,j\in\mathcal{U}_t\}$, where $\mathcal{U}_r$ and $\mathcal{U}_t$ are respectively the sets of indices of the dominant receive and transmit beam directions, $\mathbf{x}_V$ is generated as
\begin{align}
[\mathbf{x}_V(j)]=\left\{
\begin{aligned}
&\widetilde{\mathbf{x}}_V,~~~~j\in\bar{\mathcal{U}}_t,\\
&\mathbf{n}_{AN},~~j\notin\mathcal{U}_t,\\
&\mathbf{0},~~~~~~j\in\mathcal{U}_t, j\notin\bar{\mathcal{U}}_t,
\end{aligned}\right.
\label{xv}
\end{align}
where $\mathbf{n}_{AN} \in \mathbb{C}^{(N_t-L_t)\times 1}$ is the AN and $\bar{\mathcal{U}}_t$ is a subset of $\mathcal{U}_t$ with size $M_t$.
Considering that only $\mathcal{U}_t$ is available at Alice, it is natural to randomly select a subset of $\bar{\mathcal{U}}_t$ from $\mathcal{U}_t$.
Then, the transmit signal, $\mathbf{x} \in \mathbb{C}^{N_t\times 1}$, is finally obtained as
\begin{align}
\mathbf{x}=\mathbf{A}_T\mathbf{x}_V,
\label{x}
\end{align}
where $\mathbf{A}_T$ is the analog precoder.
The transmit power of $\mathbf{x}$ is limited by $P$, i.e., $\mathrm{Tr}\left(\mathbb{E}\left\{\mathbf{x} \mathbf{x}^H\right\}\right)=\mathrm{Tr}\left(\mathbb{E}\left\{\mathbf{x}_V \mathbf{x}_V^H\right\}\right)=\mathrm{Tr}\left(\mathbb{E}\left\{\widetilde{\mathbf{x}}_V \widetilde{\mathbf{x}}_V^H\right\}\right)+\mathrm{Tr}\left(\mathbb{E}\left\{\mathbf{n}_{AN} \mathbf{n}_{AN}^H\right\}\right)=P$.
Denoting by $\phi$ the proportion of power allocated for confidential signals, we have
\begin{align}
\mathrm{Tr}\left(\mathbb{E}\left\{\widetilde{\mathbf{x}}_V \widetilde{\mathbf{x}}_V^H\right\}\right)&=\phi P,
\label{P_s}
\\
\mathrm{Tr}\left(\mathbb{E}\left\{\mathbf{n}_{AN} \mathbf{n}_{AN}^H\right\}\right)&=(1-\phi)P.
\label{P_AN}
\end{align}

At Bob, the received signal is given by
\begin{align}
\mathbf{y}=\tilde{\mathbf{G}}\mathbf{x}+\mathbf{n},
\label{y}
\end{align}
where $\mathbf{n} \sim \mathcal{CN}(\mathbf{0},\sigma_n^2\mathbf{I}_{N_r})$ is the thermal noise.
Applying the receive analog decoder $\mathbf{A}_R^H$ and substituting \eqref{G} and \eqref{x} into \eqref{y} yields
\begin{align}
\mathbf{y}_V&=\mathbf{A}_R^H \mathbf{y}
=\mathbf{G}\mathbf{x}_V+\mathbf{n}_V,
\label{yv}
\end{align}
where we define $\mathbf{n}_V \triangleq \mathbf{A}_R^H\mathbf{n} \sim \mathcal{CN}(\mathbf{0},\sigma_n^2\mathbf{I}_{N_r})$ because $\mathbf{A}_R^H$ is a unitary matrix. 
Since there are only $M_r$ RF chains, we adopt a beam selector, represented by a selection mask $\bar{\mathcal{U}}_r\subseteq \mathcal{U}_r$, to randomly choose $M_r$ from the $L_r$ dominant receive beams.
From \eqref{xv} and \eqref{yv}, we obtain the signals for detection as
\begin{align}
\widetilde{\mathbf{y}}_V&=[\mathbf{y}_V(i)]_{i\in\bar{\mathcal{U}}_r}
=\bar{\mathbf{G}}\widetilde{\mathbf{x}}_V + \hat{\mathbf{G}} \mathbf{n}_{AN} + \widetilde{\mathbf{n}}_V,
\label{yv_til}
\end{align}
where $\bar{\mathbf{G}}=\left[ \mathbf{G}(i,j)\right]_{i\in \bar{\mathcal{U}}_r,j\in \bar{\mathcal{U}}_t}\in \mathbb{C}^{M_r\times M_t}$ and
$\hat{\mathbf{G}}=\left[ \mathbf{G}(i,j)\right]_{i\in \bar{\mathcal{U}}_r,j\notin \mathcal{U}_t}\in \mathbb{C}^{M_r\times (N_t-L_t)}$ are submatrices of $\mathbf{G}$, and $\widetilde{\mathbf{n}}_V=\left[\mathbf{n}_V(i)\right]_{i\in \bar{\mathcal{U}}_r}\in \mathbb{C}^{M_r\times 1}$.

The expression in \eqref{yv_til} implies that the information carrying signal, $\widetilde{\mathbf{x}}_V$, is transmitted through the dominant beams in $\bar{\mathbf{G}}$ while the AN, $\mathbf{n}_{AN}$, is sent via the nondominant beams.
From \cite{Raghavan}, it has been shown that the channel coefficients in $\bar{\mathbf{G}}$ are approximately equal to the sum of the complex gains of a set of physical paths.
When there are a sufficiently large number of paths, each entry of $\bar{\mathbf{G}}$ tends to behave as a complex Gaussian random variable due to the Central Limit Theorem \cite{channel}.
Assuming that distinct channel coefficients correspond to approximately disjoint subsets of paths and that the path gains are statistically independent, we assume that the entries of $\bar{\mathbf{G}}$ are statistically independent.
For example, the channel measurements in \cite{channel_measure1} showed an average number of 10 distinct clusters and 9 rays in each cluster for a practical 60 GHz mmWave communication scenario.
Thus, it is reasonable to approximate the elements of $\bar{\mathbf{G}}$ by zero-mean independent complex Gaussian variables \cite{Recon0}, \cite{Recon}, as further validated by the experimental measurement results in \cite{channel_measure2}.
Without loss of generality, we assume that long-term power control is employed to compensate for the large-scale fading of Bob such that the entries of $\bar{\mathbf{G}}$ have unit variance.
Similarly, the entries of $\hat{\mathbf{G}}$ are modeled by independent complex Gaussian variables with zero mean and variance $\eta$.
In general, the nondominant channel coefficients in $\hat{\mathbf{G}}$ experience a small gain in the power domain. Since the variance of the dominant coefficients in $\bar{\mathbf{G}}$ is modeled unit, we assume that $\eta\in(0,1)$.


At Eve, the receive signal vector can be expressed as
\begin{align}
\mathbf{y}_{E,V}&=\mathbf{A}_E^H\tilde{\mathbf{H}}\mathbf{x}+\mathbf{A}_E^H\mathbf{n}_E
=\mathbf{H}\mathbf{x}_V+\mathbf{A}_E^H\mathbf{n}_E,
\label{yev}
\end{align}
where $\mathbf{n}_E \sim \mathcal{CN}(\mathbf{0},\sigma_e^2\mathbf{I}_{N_e})$ denotes the thermal noise at Eve.
To guarantee secure transmission in the worst case, $\sigma_e^2$ is assumed to be small enough so that $\mathbf{n}_E$ can be ignored in the sequel \cite{Secure_Zhu_2}.
Given $M_e$ RF chains at Eve, a beam selector $\bar{\mathcal{E}}_r\subseteq \mathcal{E}_r$ is used to randomly choose $M_e$ dominant receive beams from $\mathcal{E}_r$, where $\mathcal{E}_r=\{i|(i,j)\in\mathcal{E}\}$ contains all the dominant receive beams at Eve.
Substituting \eqref{xv} into \eqref{yev}, the low-dimensional signal vector after beam selection equals
\begin{align}
\widetilde{\mathbf{y}}_{E,V}&=[\mathbf{y}_{E,V}(i)]_{i\in\bar{\mathcal{E}}_r}
=\bar{\mathbf{H}} \widetilde{\mathbf{x}}_V +\hat{\mathbf{H}} \mathbf{n}_{AN},
\label{yev_til}
\end{align}
where 
$\bar{\mathbf{H}}=\left[ \mathbf{H}(i,j)\right]_{i\in \bar{\mathcal{E}}_r,j\in\bar{\mathcal{U}}_t}\in \mathbb{C}^{M_e\times M_t}$ and $\hat{\mathbf{H}}=\left[ \mathbf{H}(i,j)\right]_{i\in \bar{\mathcal{E}}_r,j\notin\mathcal{U}_t}\in \mathbb{C}^{M_e\times (N_t-L_t)}$.
Similarly, we assume that $\left[ \mathbf{H}(i,j)\right]_{i\in \bar{\mathcal{E}}_r,j\in \mathcal{E}_t}\sim \mathcal{CN}(\mathbf{0},\mathbf{I})$ and $\left[ \mathbf{H}(i,j)\right]_{i\in \bar{\mathcal{E}}_r,j\notin \mathcal{E}_t} \sim \mathcal{CN}(\mathbf{0},\eta\mathbf{I})$, where $\mathcal{E}_t=\{j|(i,j)\in\mathcal{E}\}$ contains all the dominant transmit beams of $\mathbf{H}$.

Unlike Bob, the confidential signals in \eqref{yev_til} are likely to be allocated on nondominant beams for the Eve's channel, while AN would likely spread over dominant beams.
This mismatch between the channel sparsity patterns $\mathcal{U}$ and $\mathcal{E}$ can be exploited to degrade Eve's capacity.

\subsection{An Illustrative Example}

\begin{figure*}[tb]
\centering
\subfigure[A typical example with $M_t=M_r=M_e=L_t=L_r=L_e$.]{
\begin{minipage}[t]{0.588\linewidth}
\centering
\includegraphics[width=1\linewidth]{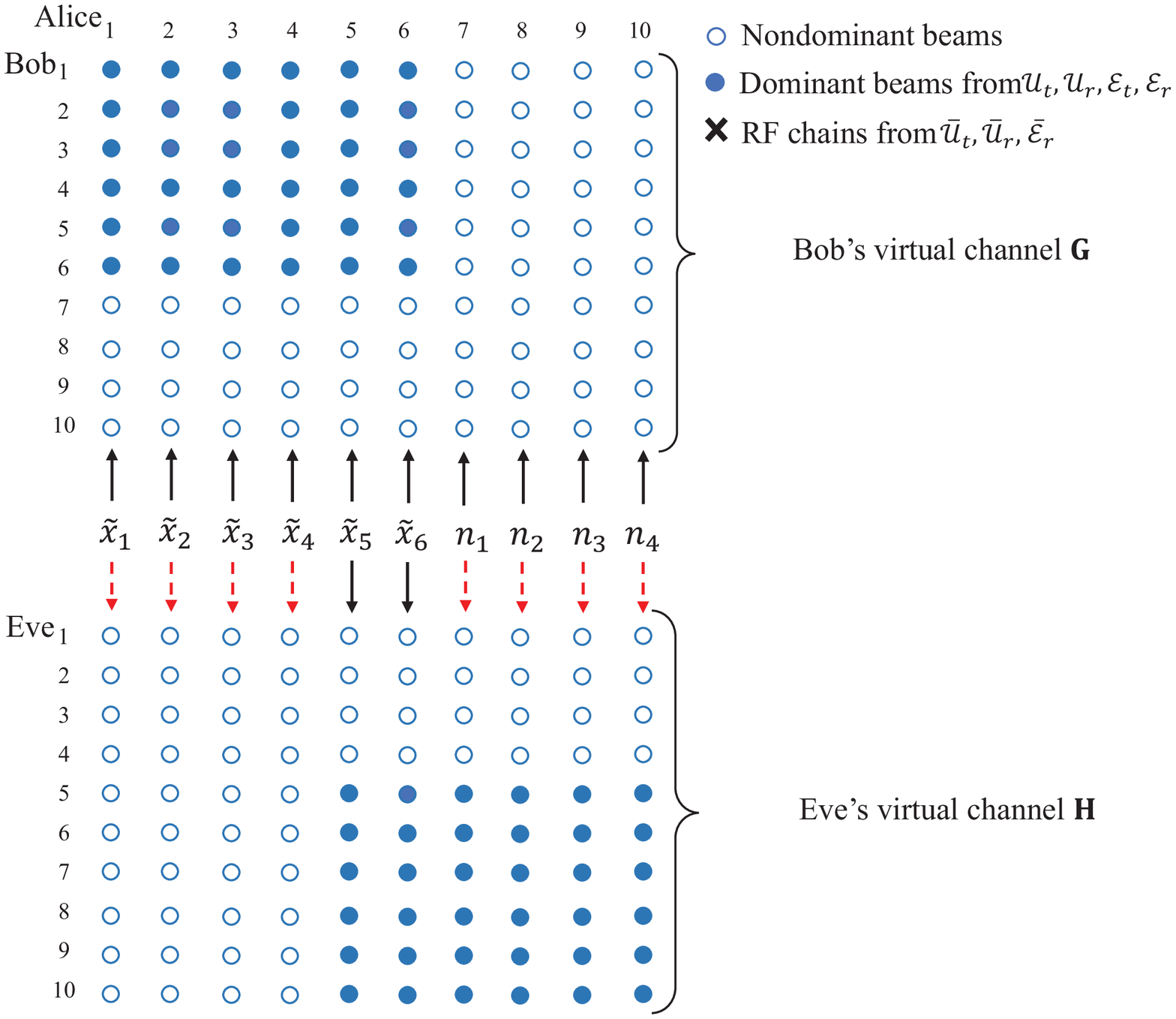}
\label{Channel_simple}
\end{minipage}}
\subfigure[A general illustration of virtual channels.]{
\begin{minipage}[t]{0.35\linewidth}
\centering
\includegraphics[width=1\linewidth]{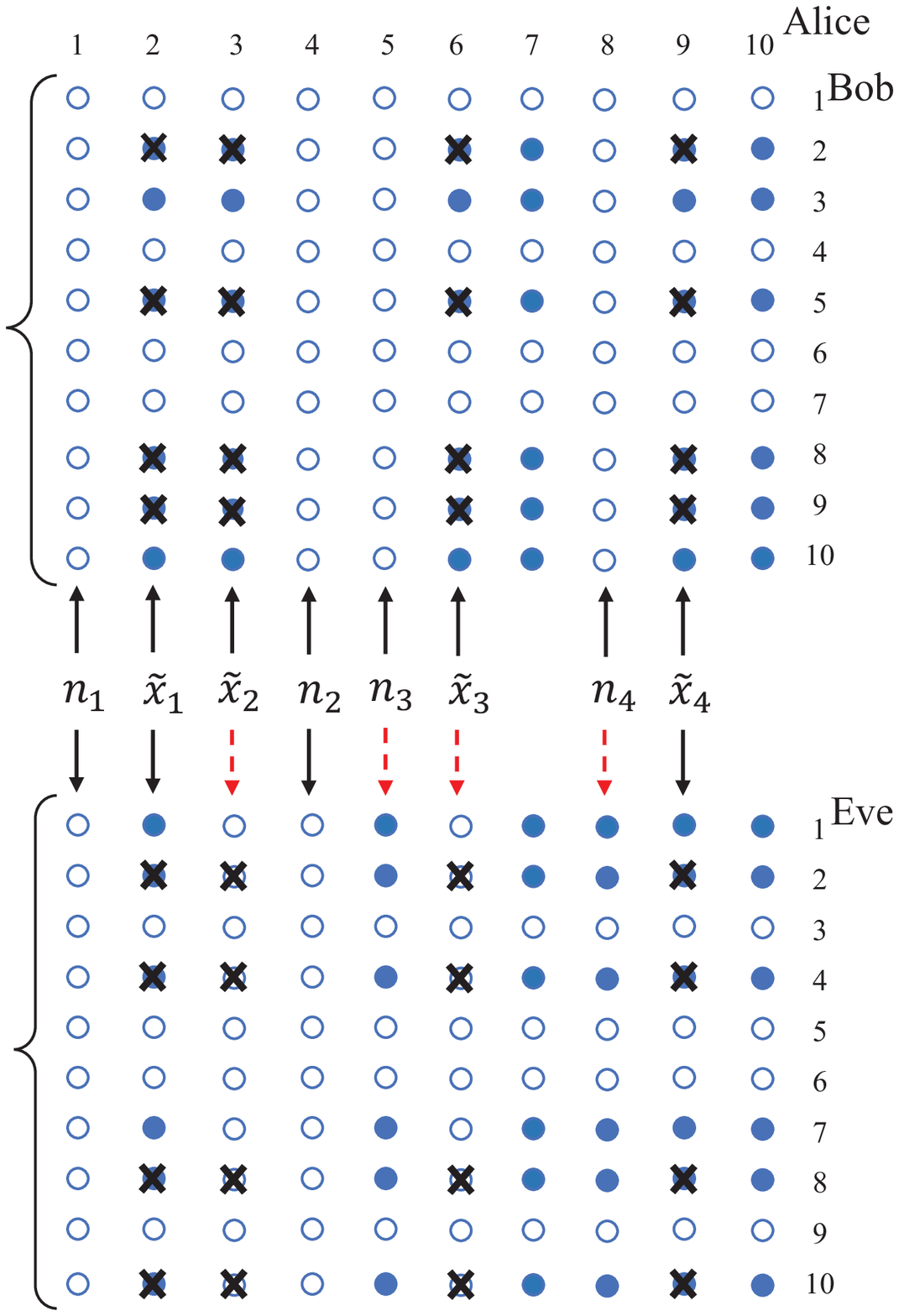}
\label{Channel_general}
\end{minipage}}
\caption{Confidential signals and AN on low-dimensional virtual channels with $N_t=N_r=N_e=10$.}
\label{Channel}
\end{figure*}

Fig. \ref{Channel} shows two examples of the proposed scheme. In Fig.~\ref{Channel_simple}, we present an example of a special case where the number of RF chains is exactly the same as the number of dominant beams in the channel, i.e., $M_t=M_r=M_e=L_t=L_r=L_e=6$ for illustration.
At Alice, the dominant beams are used for transmitting confidential signals, $\tilde{x}_k,~k\in\{1,...,6\}$, to Bob.
The AN, $n_k,~k\in\{1,...,4\}$, is sent through the nondominant beams of $\mathbf{G}$.
For Eve, mismatches exist because the sparsity patterns of $\mathbf{G}$ and $\mathbf{H}$ are distinct (the mismatches are indicated in Fig. \ref{Channel_simple} by dashed arrows). Specifically, signals $\{\tilde{x}_1,...,\tilde{x}_4\}$ are transmitted over nondominant beams of $\mathbf{H}$, therefore it is challenging for Eve to wiretap the information signals.
On the other hand, AN $\{n_1,...,n_4\}$ is sent through the dominant beams of $\mathbf{H}$, which causes significant interference to Eve but not to Bob.
Using this proposed scheme, the capacity of Eve's channel degrades significantly resulting in a potential increase in the secrecy rate.

In general, dominant beams can be dispersive in the angular directions. Also, the number of RF chains for signal transmission and reception can be different from the number of dominant beams.
Fig.~\ref{Channel_general} shows a general scenario, e.g., with $M_t=M_r=M_e=4$ and $L_t=L_r=L_e=6$.
The beam selectors are $\bar{\mathcal{U}}_r=\{2,5,8,9\}\subset\mathcal{U}_r=\{2,3,5,8,9,10\}$, $\bar{\mathcal{U}}_t=\{2,3,6,9\}\subset\mathcal{U}_t=\{2,3,6,7,9,10\}$, $\bar{\mathcal{E}}_r=\{2,4,8,10\}\subset\mathcal{E}_r=\{1,2,4,7,8,10\}$, and $\mathcal{E}_t=\{2,5,7,8,9,10\}$.
Unlike Fig.~\ref{Channel_simple}, neither signal nor noise is transmitted over the $7$th and $10$th dominant beams because of the limited number of RF chains.
The confidential signals and AN are allocated based on the sparsity pattern of $\mathbf{G}$, which is unknown to Eve.
Therefore, the sparsity acts analogously to a secret key which is beneficial for secure communication.

\section{Secrecy Rate Analysis}

In this section, we analyze the ergodic achievable secrecy rate of the considered mm-Wave massive MIMO system.
From \cite{Sec_broadcast}, \cite{Sec_sumrates}, the secrecy capacity is given by
\begin{align}
\label{Cs}
C_S=\max\limits_{\mathbf{s}\rightarrow\widetilde{\mathbf{x}}_V\rightarrow\widetilde{\mathbf{y}}_V, \widetilde{\mathbf{y}}_{E,V}} I\Big(\mathbf{s};\widetilde{\mathbf{y}}_V\Big) - I\Big(\mathbf{s};\widetilde{\mathbf{y}}_{E,V}\Big),
\end{align}
where $I(\cdot;\cdot)$ denotes the mutual information between two random variables.
The secrecy capacity $C_S$ is given by maximizing over all joint distributions such that a Markov chain $\mathbf{s}\rightarrow\widetilde{\mathbf{x}}_V\rightarrow\widetilde{\mathbf{y}}_V, \widetilde{\mathbf{y}}_{E,V}$ is formed.
Assuming that only statistical CSI of $\mathbf{G}$ is known to Alice, the digital precoder is chosen as $\mathbf{W}=\sqrt{\frac{\phi P}{M_t}}\mathbf{I}_{M_t}$ with $M_t=N_s$ \cite{Recon0}, \cite{Recon}, guaranteeing the power constraint in \eqref{P_s}.
For the sake of tractability, we then follow the definition of ergodic achievable secrecy rate in \cite[Lemma~1]{Secure_Zhu_1} and derive an expression for the secrecy rate under the assumption of large $N_t$.


\begin{theorem}
\label{theorem_Rs}
Under the assumptions of Gaussian transmit signalling, $\mathbf{s}$, and large antenna arrays, the ergodic achievable secrecy rate of Bob is
\begin{align}
R_{S}=[R_{U}-C_{E}]^+,
\label{Rs}
\end{align}
where $R_{U}$ is the ergodic achievable rate of Bob given by
\begin{align}
R_{U}=M_t \log_2\left[ 1+\frac{M_r \phi P}{M_t \sigma_n^2} \left(1-\frac{\mathcal{F}(\alpha,\beta)}{4\alpha\beta}\right) \right],
\label{Cu}
\end{align}
and $C_{E}$ is the ergodic capacity of Eve
\begin{align}
C_{E}&=\frac{L_t M_t}{N_t}\log_2\left( 1\!+\!\frac{\phi(N_t-L_t)M_e}{(1-\phi)M_t(a-M_e)b} \right)
\!+\! \frac{(N_t-L_t)M_t}{N_t}\log_2\left( 1\!+\!\frac{\phi(N_t-L_t)M_e\eta}{(1-\phi)M_t(a-M_e)b} \right),
\label{Ce}
\end{align}
where we define
\begin{align}
\label{alpha}
\alpha&\triangleq\frac{(1-\phi)\eta P}{\sigma_n^2},~~~~~~~~~~~~~~~~~~~~~~~~~~~~~~~~
\end{align}
\begin{align}
\label{beta}
\beta&\triangleq\frac{M_r}{N_t-L_t},~~~~~~~~~~~~~~~~~~~~~~~~~~~~~~~~~~~
\end{align}
\begin{align}
\label{func_F}
\mathcal{F}(x,y)&\triangleq \left( \sqrt{x(1+\sqrt{y})^2+1}- \sqrt{x(1-\sqrt{y})^2+1}\right)^2 ,
\end{align}
\begin{align}
\label{a}
a&\triangleq\frac{\left[L_t+\eta(N_t-L_t)\right]^2(N_t-L_t)}{N_t\left[L_t+\eta^2(N_t-L_t)\right]},~~~~~~~~~
\end{align}
\begin{align}
\label{b}
b&\triangleq\frac{L_t+\eta^2(N_t-L_t)}{L_t+\eta(N_t-L_t)}.~~~~~~~~~~~~~~~~~~~~~~~
\end{align}
\end{theorem}
\begin{proof}
See Appendix~\ref{proof_theorem_Rs}.
\end{proof}

From \eqref{Rs}, at first glance the effects of $L_t$, $\eta$, $\phi$ on the secrecy rate performance are quite complicated.
In order to obtain useful insights for system design, we further characterize the asymptotic behaviour of $R_s$. To this end, we first present some asymptotic results for $R_U$ and $C_E$ in the following three propositions.
For notational brevity, we introduce
\begin{align}
\rho\triangleq\frac{L_t}{N_t} 
\label{rho}
\end{align}
to represent the proportion of dominant beams from the point-of-view of Alice.


\begin{proposition}
At low SNR, the ergodic achievable rate of Bob in \eqref{Cu} can be expressed as
\begin{align}
R_{U}\approx M_t \log_2\left( 1+\frac{M_r \phi P}{M_t \sigma_n^2}  \right) \triangleq R_U^\mathrm{L}.
\label{Cu_low}
\end{align}
\label{proposition_Cu_L}
\end{proposition}
\begin{proof}
We consider a low SNR case, i.e., $P\ll\sigma_n^2$, which is a typical scenario in many mm-Wave massive MIMO applications \cite{Heath2016mmWave}, \cite{Ferrante2017mmWave}. Using \eqref{func_F}, it follows that
\begin{align}
\frac{\mathcal{F}(\alpha,\beta)}{4\alpha\beta}&=\frac{1}{4\beta} \left( \sqrt{\left(1\!+\!\sqrt{\beta}\right)^2+\frac{1}{\alpha}}\!-\! \sqrt{\left(1\!-\!\sqrt{\beta}\right)^2+\frac{1}{\alpha}}\right)^2
\approx 0,
\label{F_low}
\end{align}
where we use $\alpha\ll 1$ from \eqref{alpha} for $1-\phi<1$, $\eta<1$, and $P\ll\sigma_n^2$.
By substituting \eqref{F_low} into \eqref{Cu}, the achievable rate of Bob in \eqref{Cu_low} is directly obtained.
\end{proof}

\textbf{Remark 1.}
According to \eqref{xv}, confidential signals are transmitted over $M_t$ dominant beams while AN is sent over $N_t-L_t$ nondominant beams from Alice to Bob.
The sparser the channel is, i.e., for a smaller $L_t$, AN is transmitted on a larger set of nondominant beams.
Then, the AN leakage to Bob statistically becomes larger which decreases the rate of Bob as observed in \eqref{Cu}.
However, this impairment is relatively weak in the low SNR regime.
Observing \eqref{Cu_low}, the ergodic achievable rate of Bob increases with the number of RF chains, $M_r$, and the SNR, $\frac{P}{\sigma_n^2}$, regardless of the number of dominant beams, $L_t$, because the achievable rate is dominated by thermal noise, instead of AN in the low SNR regime.

\begin{proposition}
At high SNR, the ergodic achievable rate of Bob in \eqref{Cu} can be expressed as
\begin{align}
R_{U}\approx M_t \log_2\left( 1+\frac{M_r \phi }{M_t(1-\phi)\eta \left(1-\frac{M_r}{N_t(1-\rho)}\right)}  \right)
\triangleq R_U^\mathrm{H}.
\label{Cu_high}
\end{align}
\label{proposition_Cu_H}
\end{proposition}
\begin{proof}
Using \eqref{func_F} with $\nu\triangleq\frac{1}{\alpha}$, we define
\begin{align}
\mathcal{G}(\nu,\beta) &\triangleq \frac{\mathcal{F}(1/\nu,\beta)}{4\beta/\nu}
=\frac{1}{4\beta} \left( \sqrt{\left(1\!+\!\sqrt{\beta}\right)^2+\nu}\!-\! \sqrt{\left(1\!-\!\sqrt{\beta}\right)^2+\nu}\right)^2.
\label{G_fun}
\end{align}
For high SNR with $P\gg\sigma_n^2$, i.e., $\nu\ll1$ using \eqref{alpha}, applying a Taylor series expansion yields
\begin{align}
\mathcal{G}(\nu,\beta)&=\mathcal{G} |_{\nu=0}+\frac{\partial \mathcal{G}}{\partial \nu} \Big|_{\nu=0} \nu+\textrm{o}(\nu)
=1-\frac{\nu}{1-\beta}+\textrm{o}(\nu),
\label{G_Taylor}
\end{align}
where $\textrm{o}(\nu)$ is an insignificant higher-order term with respect to (w.r.t.) $\nu$ and we set $\nu=0$ in
\begin{align}
\frac{\partial \mathcal{G}}{\partial \nu}=\frac{-\left( \sqrt{\left(1\!+\!\sqrt{\beta}\right)^2+\nu}\!-\! \sqrt{\left(1\!-\!\sqrt{\beta}\right)^2+\nu}\right)^2}{4\beta\sqrt{\left(1\!+\!\sqrt{\beta}\right)^2+\nu}\sqrt{\left(1\!-\!\sqrt{\beta}\right)^2+\nu}}.
\end{align}

By substituting \eqref{G_fun} and \eqref{G_Taylor} into \eqref{Cu}, it follows that
\begin{align}
R_U= M_t \log_2\left( 1+\frac{M_r \phi P }{M_t \sigma_n^2 } \left[\frac{\nu}{1-\beta}-\textrm{o}(\nu)\right]  \right)
\approx M_t \log_2\left( 1+\frac{M_r \phi P \nu}{M_t \sigma_n^2(1-\beta)}  \right).
\label{Cu_ter}
\end{align}
We complete the proof by substituting \eqref{alpha}, \eqref{beta}, and \eqref{rho} into \eqref{Cu_ter}.
\end{proof}

\textbf{Remark 2.}
From \eqref{Cu_high}, $R_U^\mathrm{H}$ increases with $\rho$, which coincides with the intuition that the presence of more dominant beams is beneficial to the rate of Bob without considering security. On the other hand, it is observed from \eqref{Cu_high} that $R_U^\mathrm{H}$ decreases with $\eta$. A larger $\eta$ means that the nondominant beams degrade less significantly compared to the dominant ones, which results in exceeding large AN interference to Bob.
Unlike the low SNR case in \eqref{Cu_low}, $R_U^\mathrm{H}$ is eventually limited by an upper bound with increasing SNR because the interference caused by AN dominates the thermal noise at high SNR.


The above two propositions characterize the asymptotic rate of Bob. In order to analyze the asymptotic behaviour of the ergodic achievable secrecy rate, we also need the following proposition that reveals the asymptotic ergodic capacity of Eve.

\begin{proposition}
Under the assumption of large antenna arrays with $N_t\gg M_e$, the ergodic capacity of Eve in \eqref{Ce} can be expressed as
\begin{align}
C_E\!=\!& \underbrace{ M_t \rho \log_2\!\left(\!\! 1+\!\frac{\phi M_e}{(\!1\!-\!\phi\!)M_t\left[\rho+\eta(1-\rho)\right]} \!\right) }_{T_1}
\!+\!\underbrace{ M_t(1\!-\!\rho) \log_2\!\!\left(\!\! 1\!\!+\!\!\frac{\phi M_e\eta}{(\!1\!-\!\phi\!)M_t\!\left[\rho+\eta(1-\rho)\right]} \!\right)\!}_{T_2}.
\label{Ce_app}
\end{align}
\label{proposition_Ce}
\end{proposition}
\begin{proof}
Substituting \eqref{a}, \eqref{b}, and \eqref{rho} into \eqref{Ce}, we have
\begin{align}
C_E\!=&M_t \rho \log_2\!\left(\! 1+\!\frac{\phi M_e}{(1-\phi)M_t\left(\rho+\eta(1-\rho)-\frac{M_e[\rho+\eta^2(1-\rho)]}{N_t(1-\rho)[\rho+\eta(1-\rho)]}\right)} \right)
\nonumber\\
&+M_t(1-\rho)\log_2\!\left(\! 1\!+\!\frac{\phi M_e\eta}{(1-\phi)M_t\left(\rho+\eta(1-\rho)-\frac{M_e[\rho+\eta^2(1-\rho)]}{N_t(1-\rho)[\rho+\eta(1-\rho)]}\right)} \right)
\label{Ce_ter}
.
\end{align}
For $N_t\gg M_e$, it follows that $\frac{M_e[\rho+\eta^2(1-\rho)]}{N_t(1-\rho)[\rho+\eta(1-\rho)]} \rightarrow 0$.
Substituting this into \eqref{Ce_ter} yields \eqref{Ce_app}.
\end{proof}

In \eqref{Ce_app}, the term $T_1$ corresponds to the capacity component contributed by the dominant transmit beams for Eve's channel, while $T_2$ corresponds to that of the nondominant beams.
The effect of $\rho$ within the logarithmic function represents the impact of channel sparsity on the equivalent SNR.
From the denominator in \eqref{Ce_app}, $(1-\phi)M_t\left[\rho+\eta(1-\rho)\right]$, we observe that a proportion, $\rho$, of AN is transmitted over the dominant beams while the rest $1-\rho$ is sent through the nondominant beams which experience a degraded channel gain $\eta$.

Moreover, from \eqref{Ce_app}, we observe that $C_E$ first decreases and then increases with $\rho$. 
This indicates that compared to non-sparse channels with $\rho=1$, the sparsity degrades the ergodic capacity of Eve because the confidential signals are likely to be sent through the nondominant beams of the sparse channel.
However, when the channel is even sparser with a small $\rho$, i.e., $L_t\rightarrow M_t$, the ergodic capacity of Eve increases slightly.
This is because there are insufficient dominant beams for transmitting AN which reduces the effectiveness of the AN for combating the channel of Eve.


\section{Effects of Sparsity on Secrecy Rate}

This section presents two tight bounds for the ergodic achievable secrecy rate in \eqref{Rs} using Propositions~\ref{proposition_Cu_L}-\ref{proposition_Ce}. The derived bounds allow us to quantitatively characterize the effects of the channel sparsity on the proposed secure mm-Wave MIMO system.

\subsection{New Definitions of Statistics of Channel Sparsity}

In order to characterize the effect of channel sparsity on secrecy rate, we need to find an effective way of quantitatively measuring the sparsity. In our study, we model the degrees of spatial sparsity from two aspects.
The first one measures the sparsity in the angle domain, $\rho$, while the other measures the sparsity in the power domain, $\eta$.
The former $\rho=\frac{L_t}{N_t}$ measures the number of dominant beams $L_t\in[M_t,N_t]$, quantifying how sparse the spatial channel is.
If the channel gain of the nondominant beams is weak and can be neglected, then $\rho$ represents the available degrees of freedom in Bob's channel.
More specifically, when $\rho\rightarrow \frac{M_t}{N_t}~(L_t\rightarrow M_t)$, the channel is severely sparse and the RF chains connect with almost all the dominant beams.
On the other hand, when $\rho\rightarrow 1~(L_t\rightarrow N_t)$, the sparsity is insignificant and only a small proportion of the dominant beams are connected to the RF chains.
The other sparsity indicator $\eta\in(0,1)$ quantifies how much the nondominant beams differ from the dominant ones in the power domain.
When $\eta\rightarrow 0$, the sparsity is severe, while for $\eta\rightarrow 1$, the sparsity is much less pronounced.
From the derived results, we find it useful to define new metrics, $\chi_\mathrm{L}$ and $\chi_\mathrm{H}$, as two effective measures of the channel sparsity for analyzing the secrecy rate.
Since the two metrics are directly connected to the quantitative effect of channel sparsity on secrecy rate, we first present their definitions as well as some properties in the following.

\begin{definition}
The terms
\begin{align}
\chi_\mathrm{L}(\rho,\eta) \triangleq \eta^{\rho-1}\left[\eta+(1-\eta)\rho\right],
\label{g1}
\end{align}
and
\begin{align}
\chi_\mathrm{H}(\rho,\eta) & \triangleq \eta^{\rho-2}\left[\eta+(1-\eta)\rho\right]\left(\!1\!-\!\frac{M_r}{N_t(1-\rho)}\!\right)^{-1},
\label{g2}
\end{align}
are defined as quantitative measures of sparsity for low and high SNRs, respectively.
\label{def}
\end{definition}

The contribution of channel sparsity to the secrecy rate is determined solely by $\chi_\mathrm{L}$ and $\chi_\mathrm{H}$ for low and high SNRs, respectively.
Comparing $\chi_\mathrm{H}$ in \eqref{g2} to $\chi_\mathrm{L}$ in \eqref{g1}, we find that
\begin{align}
\chi_\mathrm{H}(\rho,\eta)=\frac{1}{\eta}\chi_\mathrm{L}(\rho,\eta),
\end{align}
which comes from the fact that $\frac{M_r}{N_t(1-\rho)}\approx 0$ for large $N_t$ and a small-to-moderate value of $\rho$.
This is reasonable because at low SNR, the achievable rate is mainly determined by the thermal noise while at high SNR the dominating interference is due to AN which experiences a channel gain of $\eta$.

In the next subsection, we will show that the ergodic achievable secrecy rate is an increasing function of $\chi_\mathrm{L}~(\chi_\mathrm{H})$ and the effects of channel sparsity on the secrecy rate are therefore completely determined by the metrics $\chi_\mathrm{L}$ and $\chi_\mathrm{H}$.

\subsection{Effects of Channel Sparsity}

Given the above definitions, in the following theorem we characterize the effects of channel sparsity on the secrecy rate by deriving two bounds for the ergodic achievable secrecy rate in \eqref{Rs} at low and high SNRs.

\begin{theorem}
\label{theorem_Rs_bound}
Assuming the existence of a powerful Eve with $M_e\gg M_t$, the upper bounds for the ergodic achievable secrecy rate at low and high SNRs are, respectively,
\begin{align}
\bar{R}_S^\mathrm{L}=&M_t \left[  \log_2\left( 1+\frac{M_r \phi P}{M_t \sigma_n^2}  \right)
-  \log_2\left(\frac{\phi M_e}{(1-\phi)M_t} \!\right)
+ \log_2 \chi_\mathrm{L} \right]^+,
\label{Rs_low_bound}
\end{align}
and
\begin{align}
\bar{R}_S^\mathrm{H}=&M_t \left[  \log_2\!\!\left(\frac{M_r }{M_e }  \!\right)
+\log_2\chi_\mathrm{H} \right]^+\!.~~~~~~~~~~~~~~~~~~~~~~~~~~~~~~~~~~~
\label{Rs_high_bound}
\end{align}
\end{theorem}

\begin{proof}
See Appendix~\ref{proof_lemma_Rs}.
\end{proof}

In \eqref{Rs_low_bound}, the first term in the brackets represents the ergodic achievable rate of Bob, the second term represents the ergodic capacity of Eve without consideration of the sparsity, and the third term captures the additional secrecy rate bonus due to exploiting the sparsity.
Similarly, the first term in the bracket in \eqref{Rs_high_bound} represents the ergodic achievable secrecy rate without sparsity and the second term represents the secrecy rate bonus due to spatial sparsity.
Using \eqref{g1} and \eqref{g2}, it can be easily verified that $\chi_\mathrm{H}>\chi_\mathrm{L}\geq 1$ for $\rho\in\left[\frac{M_t}{L_t},1\right]$ and $\eta\in(0,1)$. Hence, the secrecy rate bonus is always nonnegative.


From Theorem~\ref{theorem_Rs_bound}, it is generally difficult to give an exact elaboration on the secrecy rate bonus of the proposed secure transmission scheme.
Here, we give an intuitive description from the view of uncertainty, i.e., ``entropy'', in terms of channel sparsity.
In mmWave secure communication, it is intuitive to interpret the secrecy rate bonus of the proposed scheme as coming mainly from the mismatch between the sparsity patterns $\mathcal{U}$ and $\mathcal{E}$.
For each RF chain at Alice, the transmit beam direction may differ in the beam types for Bob and Eve, i.e., dominant beams aligned with one terminal may serve as nondominant beams for the other.
Assuming an equal probability of a match or mismatch for each pair of transmit beams, we can interpret this probability as $\frac{1}{\chi_\mathrm{L}}$ and $\frac{1}{\chi_\mathrm{H}}$ for low and high SNRs respectively.
Then, the corresponding secrecy rate bonus, $M_t\log_2(\chi_\mathrm{L})$ and $M_t\log_2(\chi_\mathrm{H})$, represents the additional uncertainty contributed by channel sparsity.


On the other hand, for fixed system parameters including the number of RF chains, system SNR, and $\phi$, Theorem~\ref{theorem_Rs_bound} implies that the ergodic achievable secrecy rate depends only on sparsity parameters $\rho$ and $\eta$.
The following theorem quantitatively characterizes the difference in secrecy rate due to channels with different levels of sparsity.

\begin{corollary}
\label{corollary_R_gap}
For two channels with respective sparsity parameters $(\rho_1,\eta_1)$ and $(\rho_2,\eta_2)$, the ergodic achievable secrecy rate gap between the two channels is, for low SNR,
\begin{align}
\Delta R_S^\mathrm{L}(\rho_1,\eta_1;\rho_2,\eta_2)=M_t\log_2\frac{\chi_\mathrm{L}(\rho_1,\eta_1)}{\chi_\mathrm{L}(\rho_2,\eta_2)},
\label{Rs_gap_low}
\end{align}
and for high SNR,
\begin{align}
\Delta R_S^\mathrm{H}(\rho_1,\eta_1;\rho_2,\eta_2)=M_t\log_2\frac{\chi_\mathrm{H}(\rho_1,\eta_1)}{\chi_\mathrm{H}(\rho_2,\eta_2)}.
\label{Rs_gap_high}
\end{align}
\end{corollary}

The accuracy of Corollary~\ref{corollary_R_gap} will be verified by numerical results in Section~\uppercase\expandafter{\romannumeral5}.
From \eqref{Rs_gap_low} and \eqref{Rs_gap_high}, the proposed metrics $\chi_\mathrm{L}$ and $\chi_\mathrm{H}$ accurately evaluate the secrecy rate gap for two different sparse channels.

\subsection{Optimal Sparsity}

In the following, we discuss the effect of the sparsity parameters, $\rho$ and $\eta$, on $\chi_\mathrm{L}$ and $\chi_\mathrm{H}$, which equivalently describes their impact on the ergodic achievable secrecy rate.

\subsubsection{Effect of $\rho$}
When the difference between the nondominant and dominant beams is insignificant for large $\eta\rightarrow 1$, the effect of the proportion of the dominant beams, i.e., $\rho$, on the secrecy rate is not pronounced. On the other hand, for small $\eta\rightarrow 0$, the nondominant beams differ significantly from the dominant ones and thus the effect of $\rho$ is pronounced.
From the definitions in \eqref{g1} and \eqref{g2}, we show that there exists a value of the sparsity parameter $\rho$ that maximizes the secrecy rate of the system.
In order to analyze the effect of $\rho$ on the secrecy rate in the angle domain, we consider the following optimization problem
\begin{align}
\max\limits_{\rho}~~ \bar{R}_S^\mathrm{L}\left(\bar{R}_S^\mathrm{H}\right)= \max\limits_{\rho}~~\chi_\mathrm{L}\left(\chi_\mathrm{H}\right),
\label{max_gL}
\end{align}
where we exploit the fact that $\bar{R}_S^\mathrm{L}$ and $\bar{R}_S^\mathrm{H}$ monotonically increase with $\chi_\mathrm{L}$ and $\chi_\mathrm{H}$, respectively.
The closed-form solution to \eqref{max_gL} is given in Lemma~\ref{lemma_rho_star}.

\begin{lemma}
The optimal $\rho^*$ which maximizes $\chi_\mathrm{L}\left(\chi_\mathrm{H}\right)$ is
\begin{align}
\rho^*=-\frac{1}{\ln\eta}-\frac{\eta}{1-\eta},
\label{rho_star}
\end{align}
and the optimal $L_t^*$ can be obtained by using the relationship between $\rho^*$ and $L_t^*$ in \eqref{rho}.
\label{lemma_rho_star}
\end{lemma}

\begin{proof}
See Appendix~\ref{proof_lemma_rho_star}.
\end{proof}

Note that the optimal $\rho^*$ in \eqref{rho_star} depends only on $\eta$. In order to characterize the effect of $\eta$ on $\rho^*$, we derive that
\begin{align}
\frac{\partial \rho^* }{\partial \eta}=\frac{\left(\eta^{-\frac{1}{2}}-\eta^{\frac{1}{2}}+\ln \eta\right)\left(\eta^{-\frac{1}{2}}-\eta^{\frac{1}{2}}-\ln \eta\right)}{\left[(1-\eta)\ln \eta\right]^2}>0,
\label{rho_eta}
\end{align}
where we use the fact that $\eta^{-\frac{1}{2}}-\eta^{\frac{1}{2}}-\ln \eta>\eta^{-\frac{1}{2}}-\eta^{\frac{1}{2}}+\ln \eta>0$ for $\eta\in(0,1)$.
This implies that $\rho^*$ increases with $\eta$. 
As the channel gain of the nondominant beams increases in the power domain, the optimal number of nondominant beams correspondingly decreases in the angle domain.
In this way, any of the nondominant beams can produce nearly the same effect on the secure communication.

\subsubsection{Effect of $\eta$}
At Bob, the AN experiences a degraded channel gain of $\eta$ over the nondominant beams while for Eve, both confidential information and AN is transmitted over nondominant beams.
In order to analyze the effect of $\eta$ on the secrecy rate, i.e., on $\chi_\mathrm{L}\left(\chi_\mathrm{H}\right)$, we focus on the derivatives of $\chi_\mathrm{L}$ and $\chi_\mathrm{H}$ w.r.t. $\eta$ as follows
\begin{align}
\frac{\partial \chi_\mathrm{L} }{\partial \eta}=-(1-\rho)\rho(1-\eta)\eta^{\rho-2}<0,~~~~~~~~~~~~~~~~~~~~~~~~~~~~~
\label{g_eta}
\end{align}
\begin{align}
\frac{\partial \chi_\mathrm{H} }{\partial \eta}=\left[\rho(\rho-2)-\eta(\rho-1)^2\right]\eta^{\rho-3}\left(\!1\!-\!\frac{M_r}{N_t(1-\rho)}\!\right)^{-1}<0,
\end{align}
where we apply \eqref{g1} and \eqref{g2} and consider $\rho<1-\frac{M_r}{N_t}$ for common sparse scenarios.
This implies that both $\chi_\mathrm{L}$ and $\chi_\mathrm{H}$ decrease slightly with $\eta$.
Given a larger $\eta$, it is easier for Eve to wiretap the confidential information even though it is transmitted over nondominant beams, leading to a lower secrecy rate.



\section{Numerical Results}

In this section, we verify the derived results including Theorems~\ref{theorem_Rs}-\ref{theorem_Rs_bound}, Propositions~\ref{proposition_Cu_L}-\ref{proposition_Ce}, and Corollary~\ref{corollary_R_gap} under various system parameters.
In the numerical simulations, we set $N_r=N_e=N_t$ and $L_r=L_e=L_t$, and denote $\gamma_0\triangleq\frac{P}{\sigma_n^2}$ as the system SNR.

\begin{figure}[tb]
\centering\includegraphics[width=0.55\textwidth]{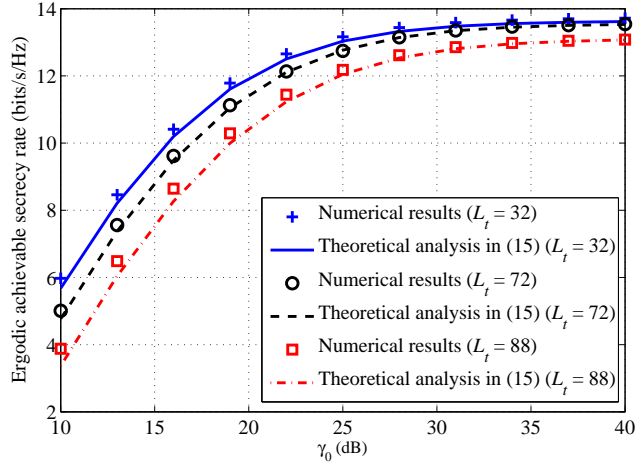}
\caption{Ergodic achievable secrecy rate versus SNR with various $L_t$ ($N_t=128$, $M_t=4$, $M_r=M_e=16$, $\phi=0.6$, and $\eta=0.1$).}
\label{Rs_SNR__Lt}
\end{figure}

\begin{figure}[tb]
\centering\includegraphics[width=0.55\textwidth]{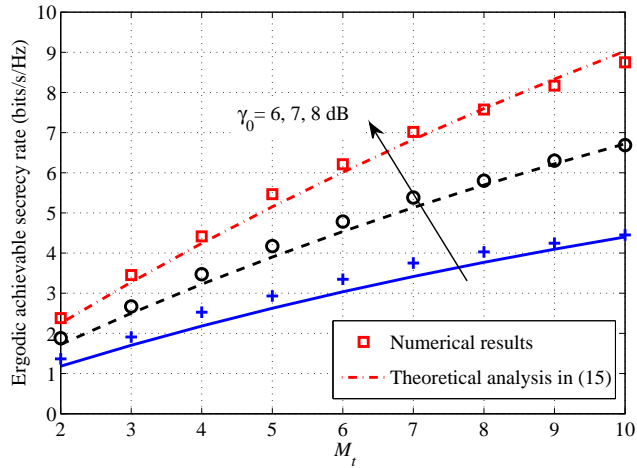}
\caption{Ergodic achievable secrecy rate versus $M_t$ with various SNRs ($N_t=256$, $L_t=28$, $M_r=M_e=20$, $\phi=0.6$, and $\eta=0.1$).}
\label{Rs_SNR__Mt}
\end{figure}

\begin{figure}[tb]
\centering\includegraphics[width=0.55\textwidth]{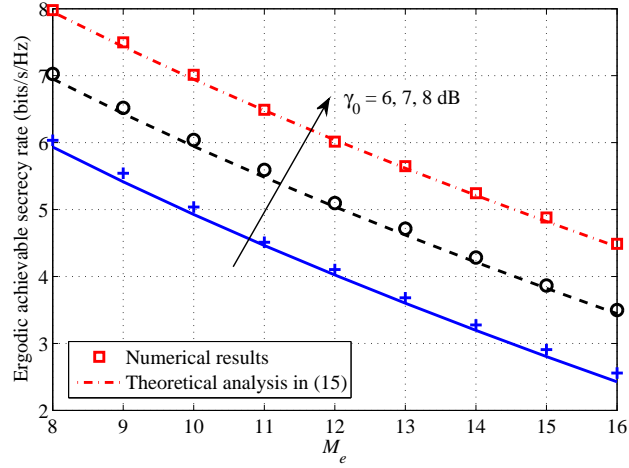}
\caption{Ergodic achievable secrecy rate versus $M_e$ with various SNRs ($N_t=256$, $L_t=28$, $M_t=4$, $M_r=16$, $\phi=0.6$, and $\eta=0.1$).}
\label{Rs_SNR__Me}
\end{figure}

\begin{figure}[tb]
\centering\includegraphics[width=0.55\textwidth]{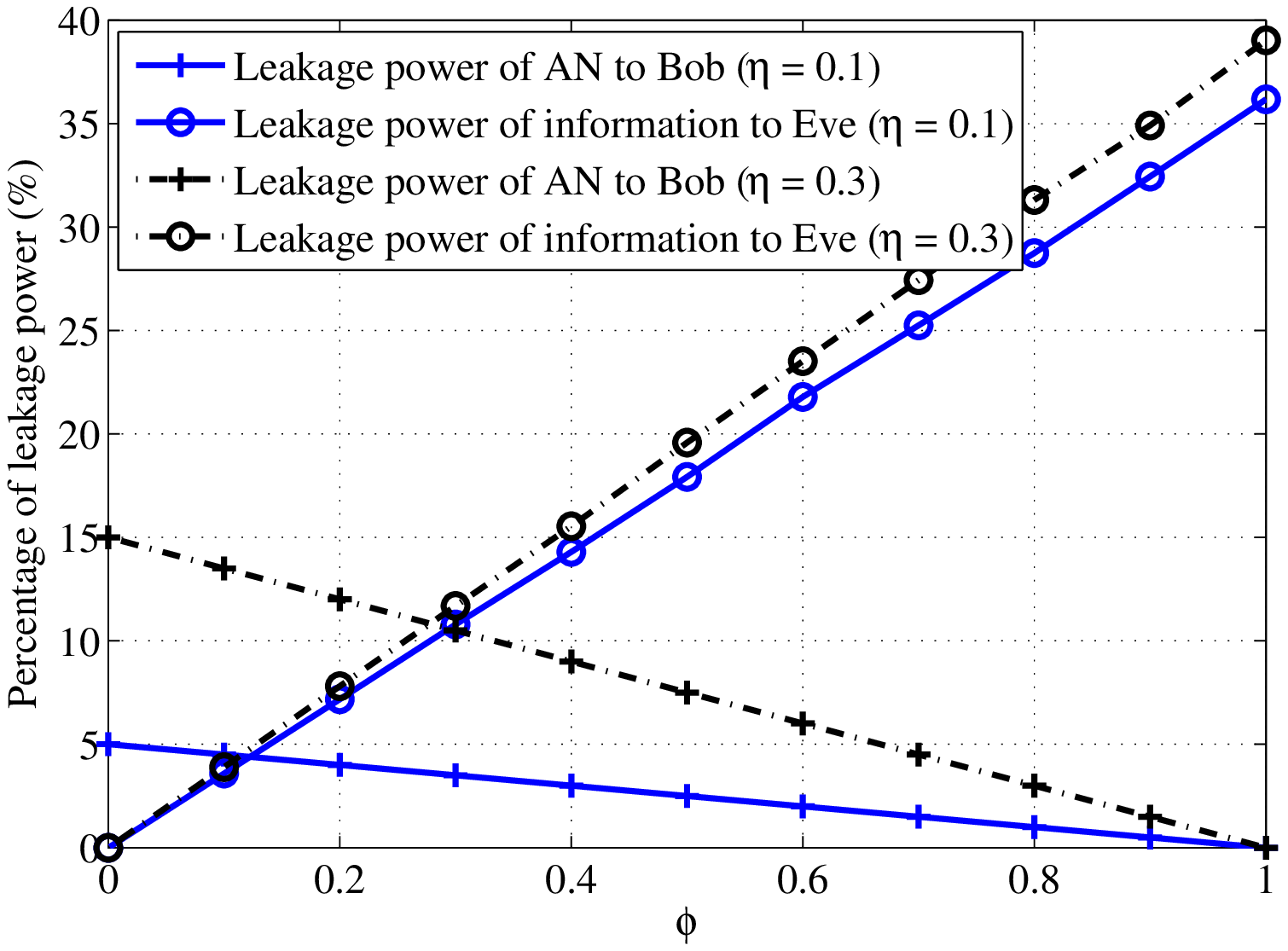}
\caption{Percentage of leakage power to transmit power $P$ versus $\phi$ ($N_t=128$, $L_t=L_r=L_e=88$, $M_t=4$, and $M_r=M_e=16$).}
\label{Powerleakage}
\end{figure}

Fig. \ref{Rs_SNR__Lt} verifies the accuracy of the derived rate expression for different values of $L_t$ versus SNR.
Dotted markers correspond to numerical results while solid lines correspond to the derived theoretical expression in \eqref{Rs}.
We observe that the ergodic achievable secrecy rate increases with $\gamma_0$ and, as expected, finally becomes saturated at high SNR due to the effect of AN.
Also, the figure shows that a smaller $L_t$ achieves higher secrecy rate, which coincides with the theoretical observation in Section~\uppercase\expandafter{\romannumeral4}.C that slight sparsity in the angle domain improves the secrecy rate.
Fig. \ref{Rs_SNR__Mt} shows the effect of $M_t$ on the ergodic achievable secrecy rate given the sparsity parameters $\rho=\frac{L_t}{N_t}=\frac{28}{256}$ and $\eta=0.1$. We observe that the secrecy rate monotonically increases with $M_t$ for the considered SNRs, $\gamma_0=6, 7$, and $8$ dB.
This is because deploying more RF chains can achieve higher beamforming gain in massive MIMO systems, although the large number of RF chains can significantly increase the circuit power consumption.
The effect of increasing the number of RF chains at Eve on the secrecy rate is shown in Fig. \ref{Rs_SNR__Me}. It is observed that the secrecy rate monotonically decreases with an increasing $M_e$ for the considered SNRs, $\gamma_0=6, 7$, and $8$ dB.
This is because more degrees of freedom are introduced by a large number of RF chains at Eve which facilitates a more efficient eavesdropping.

From Figs.~\ref{Rs_SNR__Lt}-\ref{Rs_SNR__Me}, it can be observed that the numerical results of the ergodic achievable secrecy rate is slightly higher than our theoretical analysis in Theorem~\ref{theorem_Rs}, especially when the condition $M_t\ll M_r$ holds.
The expression in Theorem~\ref{theorem_Rs} is therefore verified to be a lower bound on the ergodic achievable secrecy rate, which certainly serves as a lower bound on the ergodic secrecy capacity.


Fig. \ref{Powerleakage} shows the percentage of leakage power to total transmit power $P$, including the leakage power of AN to Bob and the leakage power of information to Eve.
Obviously, the AN leakage power decreases proportionally with the power allocation parameter $\phi$, while the information leakage power increases with $\phi$.
When the channel gain of the nondominant beams, $\eta$, increases, the AN leakage power increases significantly because the AN is leaked to Bob over these nondominant beams.
On the other hand, the information leakage power only slightly increases with $\eta$ because the information is eavesdropped by Eve over both dominant and nondominant beams.

\begin{figure}[tb]
\centering\includegraphics[width=0.55\textwidth]{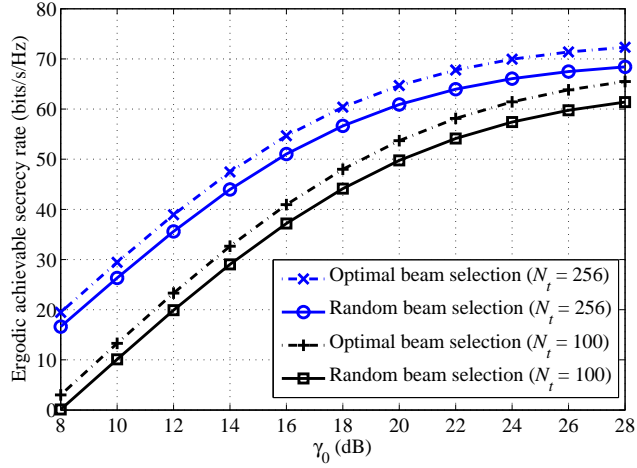}
\caption{Performance comparison between random and optimal beam selections with statistical and instantaneous CSI ($M_t=M_r=M_e=32$, $L_t=40$, $\phi=0.6$, and $\eta=0.1$).}
\label{beamselection2}
\end{figure}

In general, it is difficult to estimate the instantaneous information of the entire channel matrix due to the fact that only a limited number of RF chains are available.
Thus, statistical CSI in terms of the sparsity patterns, i.e., $\mathcal{U}$ and $\mathcal{E}$, is assumed known for beam selection.
With statistical CSI, the transmitter and receiver randomly choose a subset
of beams from the dominant channel directions in $\mathcal{U}$ and $\mathcal{E}$.
If instantaneous CSI were available at the transceiver, the optimal subset of dominant beams could be chosen to maximize the achievable secrecy rate.
Fig. \ref{beamselection2} compares the ergodic achievable secrecy rates using random and optimal beam selections.
For Alice, Bob, and Eve, a subset of 32 beams are chosen from the 40 dominant beams for the 32 RF chains exploited at each terminal.
As expected, we observe that the optimal beam selection only achieves a slightly higher secrecy rate than that of the random selection.
The performance of the optimal beam selection therefore serves as an ideal upper bound for benchmarking as acquiring the entire channel matrix with a limited number of RF chains is challenging if not impossible.

\begin{figure}[tb]
\centering\includegraphics[width=0.55\textwidth]{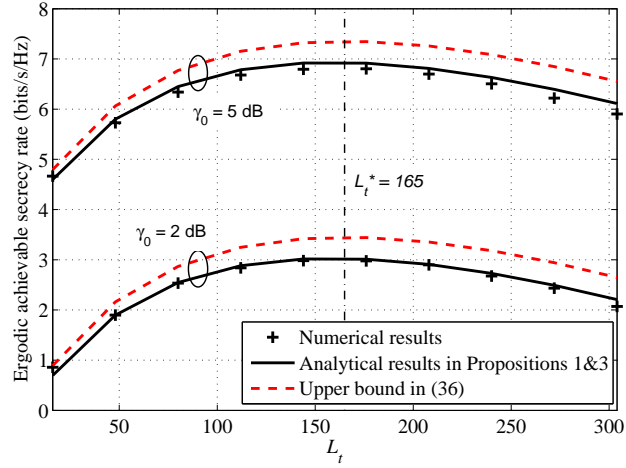}
\caption{Ergodic achievable secrecy rate and the corresponding upper bound at low SNRs ($M_t=4$, $M_r=96$, and $M_e=16$).}
\label{Rs_Lt_low}
\end{figure}

\begin{figure}[tb]
\centering\includegraphics[width=0.55\textwidth]{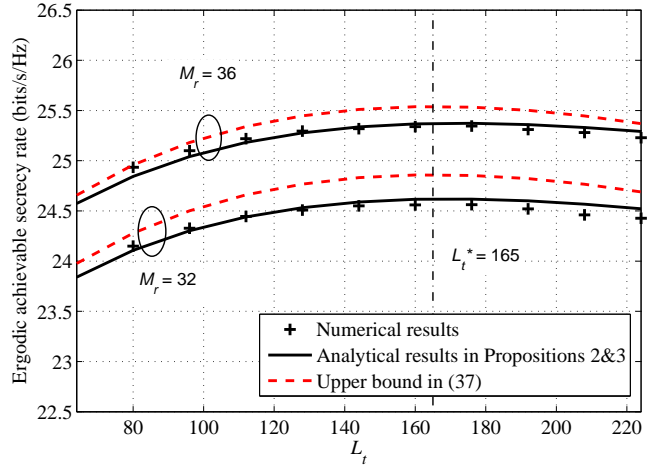}
\caption{Ergodic achievable secrecy rate and the corresponding upper bound at high SNRs ($M_t=4$, $M_e=8$, and $\gamma_0=30$ dB).}
\label{Rs_Lt_high}
\end{figure}

Fig. \ref{Rs_Lt_low} and Fig. \ref{Rs_Lt_high} compare the ergodic achievable secrecy rate and the derived upper bound versus $L_t$ at low and high SNRs, respectively.
We set $N_t=512$, $\phi=0.9$ and $\eta=0.1$.
The analytical results, which are obtained by applying Propositions~\ref{proposition_Cu_L}--\ref{proposition_Ce} to Theorem~\ref{theorem_Rs}, are shown to be accurate compared to the numerical results.
We observe that the secrecy rate first increases and then decreases with $L_t$.
As indicated in Section~\uppercase\expandafter{\romannumeral2}, the secrecy rate bonus arises from the mismatch of the sparsity patterns $\mathcal{U}_t$ and $\mathcal{E}_t$ with size $L_t$.
Given a fixed $N_t$ and as $L_t$ increases, this mismatch, in terms of the number of non-overlapped beams, first increases when $\mathcal{U}_t$ and $\mathcal{E}_t$ contain more dominant beams. Then as $L_t\rightarrow N_t$, however, both sets $\mathcal{U}_t$ and $\mathcal{E}_t$ tend to entirely overlap with each other because both contain almost all the transmit beams, and hence the mismatch vanishes.
In addition, the optimal $L_t^*=165$ obtained using Lemma~\ref{lemma_rho_star} is also verified to be accurate in the simulation results.

\begin{figure}[tb]
\centering\includegraphics[width=0.55\textwidth]{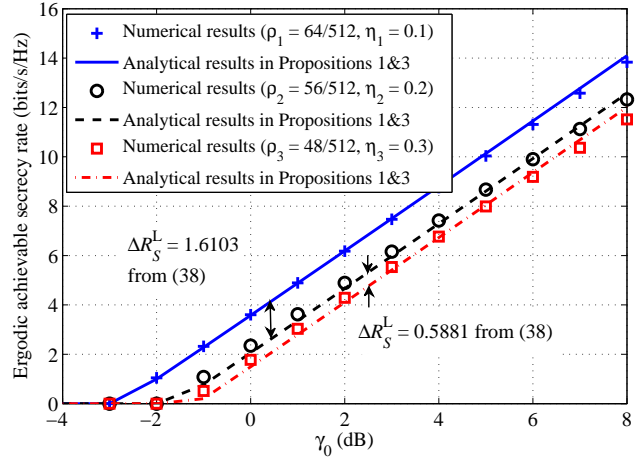}
\caption{Ergodic achievable secrecy rate gap at low SNRs ($N_t=512$, $M_t=4$, $M_r=192$ $M_e=16$, and $\phi=0.9$).}
\label{Fig_Rs_gap_low}
\end{figure}

\begin{figure}[tb]
\centering\includegraphics[width=0.55\textwidth]{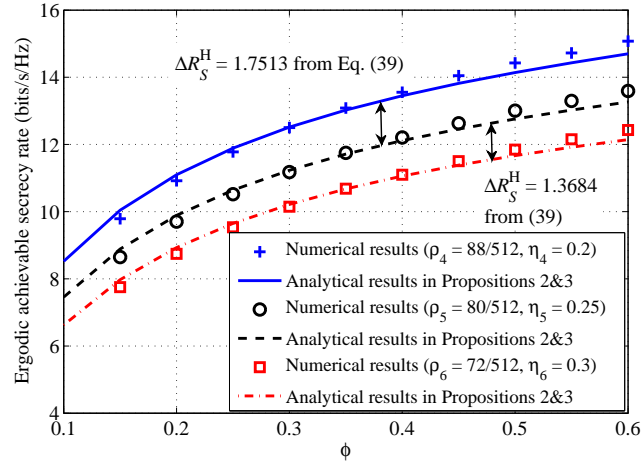}
\caption{Ergodic achievable secrecy rate gap at high SNRs ($N_t=512$, $M_t=4$, $M_r=32$ $M_e=12$, and $\gamma_0=25$ dB).}
\label{Fig_Rs_gap_high}
\end{figure}

Fig. \ref{Fig_Rs_gap_low} evaluates the ergodic achievable secrecy rate of the system under different levels of sparsity at low SNRs.
As indicated in Corollary~\ref{corollary_R_gap}, their difference in terms of secrecy rate depends only on the values of $\rho$ and $\eta$, regardless of SNR.
This is because the mismatch between the sparsity patterns of Bob's and Eve's channels are determined only by the parameters $\rho$ and $\eta$.
Using \eqref{Rs_gap_low}, secrecy rate gaps $\Delta R_S^\mathrm{L}$ of $\{(\rho_1,\eta_1), (\rho_2,\eta_2)\}$ and $\{(\rho_2,\eta_2), (\rho_3,\eta_3)\}$ are calculated respectively as $1.6103$ and $0.5881$ bits/s/Hz.
This calculation matches the numerical results in the figure by comparing with the exact ergodic achievable secrecy rate.
Fig. \ref{Fig_Rs_gap_high} shows the ergodic achievable secrecy rate gaps versus $\phi$ at high SNRs.
Similarly, the secrecy rate gaps $\Delta R_S^\mathrm{H}$ obtained from \eqref{Rs_gap_high} are $1.7513$ and $1.3684$ bits/s/Hz, which perfectly coincide with the considered cases.

\begin{figure*}[tb]
\centering
\subfigure[Contour of the metric $\chi_\mathrm{L}$ in \eqref{g1}.]{
\begin{minipage}[t]{0.48\linewidth}
\centering
\includegraphics[width=1\linewidth]{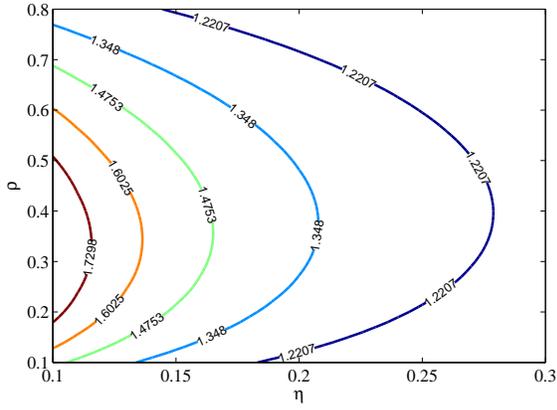}
\label{Contour_gL}
\end{minipage}}
\subfigure[Contour of the metric $\chi_\mathrm{H}$ in \eqref{g2}.]{
\begin{minipage}[t]{0.48\linewidth}
\centering
\includegraphics[width=1\linewidth]{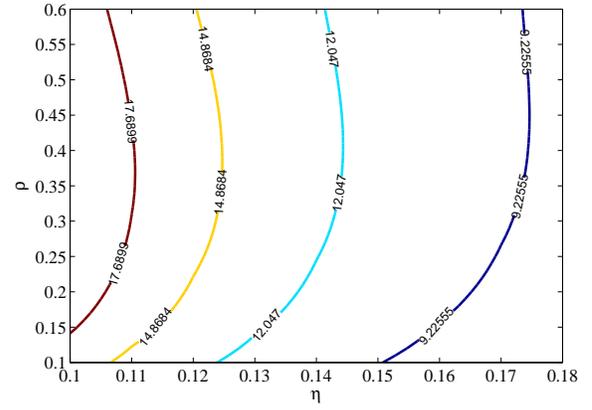}
\label{Contour_gH}
\end{minipage}}
\caption{Contours of the defined metrics measuring the secrecy performance.}
\label{Contour}
\end{figure*}

\begin{figure}[tb]
\centering\includegraphics[width=0.55\textwidth]{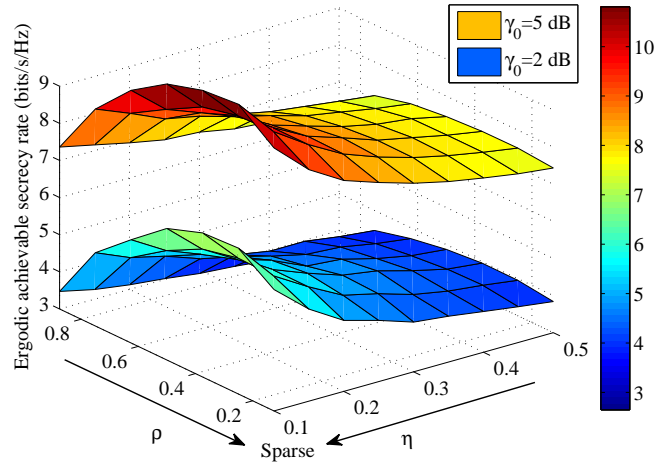}
\caption{Ergodic achievable secrecy rate versus $\rho$ and $\eta$ at low SNRs ($N_t=512$, $M_t=4$, $M_r=192$ $M_e=16$, and $\phi=0.9$).}
\label{Rs_gap_mesh_L}
\end{figure}

Fig. \ref{Contour_gL} and Fig. \ref{Contour_gH} display the contours of the metrics, $\chi_\mathrm{L}$ in \eqref{g1} and $\chi_\mathrm{H}$ in \eqref{g2}.
Larger $\chi_\mathrm{L}\left(\chi_\mathrm{H}\right)$ means better secrecy performance can be achieved.
This figure shows that both $\chi_\mathrm{L}$ and $\chi_\mathrm{H}$ decrease with $\eta$, which conforms our observation that sparsity in the power domain always contributes to secure communication.
Smaller $\eta$ means that the nondominant beams degrade more significantly compared to the dominant ones.
In this way, the mismatch between the sparsity patterns $\mathcal{U}$ and $\mathcal{E}$ becomes more pronounced, leading to higher secrecy rate.
On the other hand, there exists an optimal value, $\rho^*$, which statistically provides the optimal mismatch between the sparsity patters and achieves the largest $\chi_\mathrm{L} \left(\chi_\mathrm{H}\right)$.
The optimal $\rho^*$ increases with $\eta$ as indicated in \eqref{rho_star}.
Similar observations can be obtained from Fig.~\ref{Rs_gap_mesh_L}, which shows the secrecy rate versus $\rho$ and $\eta$ at low SNRs.
In fact, the curves in Fig.~\ref{Contour_gL} show the same performance as the contours of the 3D plot in Fig.~\ref{Rs_gap_mesh_L}, because the secrecy rate monotonically increases with $\chi_\mathrm{L}$. Comparing the two cases with $\gamma_0=2$ and $5$~dB in Fig. \ref{Rs_gap_mesh_L}, we observe that the rate bonus due to the channel sparsity is determined only by the parameters $\rho$ and $\eta$, regardless of SNR as indicated in Corollary~\ref{corollary_R_gap}.

\section{Conclusions}

In this paper, we proposed a secure communication scheme leveraging the spatial sparsity of the mm-Wave massive MIMO channel.
Confidential signals are transmitted over the dominant beams of Bob's channel while AN is injected onto the nondominant beams.
We derived two upper bounds for the ergodic achievable secrecy rate and quantitatively characterized the contribution of the sparsity as an additive secrecy rate bonus.
The rate bonus arises from uncertainty in the sparsity pattern of Bob's channel, which is unknown to Eve.
We defined two metrics for the statistics of the channel sparsity and analyzed the impact of sparsity in the angle and power domains on the secrecy rate.
In the angle domain, we derived the optimal level of sparsity that achieves the highest secrecy rate.
On the other hand, we showed that the sparsity in the power domain always benefits secure communication.
Interesting future works include further extending our current results to a general scenario with imperfect CSI at receivers for signal detection.


\begin{appendices}
\section{Preliminary Lemmas}
\label{Appendix_lemmas}

The following two lemmas and their proofs will be used in Appendix~\ref{proof_theorem_Rs} to prove Theorem~\ref{theorem_Rs}.

\begin{lemma}
\label{lemma_preliminary_1}
For $\hat{\mathbf{G}}=\left[ \mathbf{G}(i,j)\right]_{i\in \bar{\mathcal{U}}_r,j\notin \mathcal{U}_t}$ in \eqref{yv_til} whose columns follow $\mathcal{CN}(\mathbf{0},\eta\mathbf{I}_{M_r})$, we have
\begin{align}
\left[\! \frac{(\!1\!-\!\phi\!)P}{(\!N_t\!-\!L_t\!)\sigma_n^2} \hat{\mathbf{G}}\hat{\mathbf{G}}^H \!+\! \mathbf{I}_{M_r}  \!\right]^{\!-\!1}
\xrightarrow{\mathrm{a.s.}} \left(1-\frac{\mathcal{F}(\alpha,\beta)}{4\alpha\beta}\right)\mathbf{I}_{M_r},
\label{equ_lemma1}
\end{align}
where $\alpha$, $\beta$, and $\mathcal{F}(\cdot,\cdot)$ are respectively defined in \eqref{alpha}, \eqref{beta}, and \eqref{func_F} in Theorem~\ref{theorem_Rs}.
\end{lemma}
\begin{proof}
Applying the Central Limit Theorem, we have $\frac{1}{(N_t-L_t)\eta} \hat{\mathbf{G}}\hat{\mathbf{G}}^H  \xrightarrow{\mathrm{a.s.}}\mathbf{I}_{M_r} $ for large $N_t$.
This convergence is preserved for continuous matrix functions according to the Continuous Mapping Theorem \cite{theorem}, yielding that $\left(\!\frac{(\!1\!-\!\phi\!)P}{(\!N_t\!-\!L_t\!)\sigma_n^2} \hat{\mathbf{G}}\hat{\mathbf{G}}^H \!+\! \mathbf{I}_{M_r}\!\right)^{\!-\!1}$ almost surely converges to a diagonal matrix with equal diagonal entries. Combining with the result in \cite[Eq. (1.16)]{matrix},
\begin{align}
\frac{1}{M_r}\mathrm{Tr} \left(\left(\!\frac{(\!1\!-\!\phi\!)P}{(\!N_t\!-\!L_t\!)\sigma_n^2} \hat{\mathbf{G}}\hat{\mathbf{G}}^H \!+\! \mathbf{I}_{M_r}\!\right)^{\!-\!1}\right)
\xrightarrow{\mathrm{a.s.}} 1-\frac{\mathcal{F}(\alpha,\beta)}{4\alpha\beta},
\end{align}
we arrive at \eqref{equ_lemma1}.
\end{proof}

\begin{lemma}
\label{lemma_preliminary_2}
For $\hat{\mathbf{H}}=\left[ \mathbf{H}(i,j)\right]_{i\in \bar{\mathcal{E}}_r,j\notin\mathcal{U}_t}$ in \eqref{yev_til}, we have
\begin{align}
\left(\hat{\mathbf{H}}\hat{\mathbf{H}}^H\right)^{-1}
\xrightarrow{\mathrm{a.s.}}\frac{1}{(a-M_e)b}\mathbf{I}_{M_e},
\label{H2}
\end{align}
where $a$ and $b$ are respectively defined in \eqref{a} and \eqref{b} in Theorem~\ref{theorem_Rs}.
\end{lemma}

\begin{proof}
Denoting $\mathbf{h}_j$ as the $j$th column of submatrix $[\mathbf{H}(i,j)]_{i\in \bar{\mathcal{E}}_r,\forall j}$, we can decompose the matrix $\hat{\mathbf{H}}\hat{\mathbf{H}}^H$ as
\begin{align}
\hat{\mathbf{H}}\hat{\mathbf{H}}^H&=\sum\limits_{j\notin\mathcal{U}_t, j\in\mathcal{E}_t} \mathbf{h}_j\mathbf{h}_j^H + \sum\limits_{j\notin\mathcal{U}_t, j\notin\mathcal{E}_t} \mathbf{h}_j\mathbf{h}_j^H
=\mathbf{H}_1 \mathbf{H}_1^H+\mathbf{H}_2 \mathbf{H}_2^H,
\label{H2_2}
\end{align}
where $\mathbf{H}_1\triangleq\left[ \mathbf{H}(i,j)\right]_{i\in \bar{\mathcal{E}}_r, j\notin\mathcal{U}_t, j\in\mathcal{E}_t}$,
and $\mathbf{H}_2\triangleq\left[ \mathbf{H}(i,j)\right]_{i\in \bar{\mathcal{E}}_r, j\notin\mathcal{U}_t, j\notin\mathcal{E}_t}$.
Under the assumption of large $N_t$ and $L_t$, there are asymptotically $\frac{L_t(\!N_t\!-\!L_t\!)}{N_t}$ columns in $\mathbf{H}_1$ and $\frac{(\!N_t\!-\!L_t\!)^2}{N_t}$ columns in $\mathbf{H}_2$.
Then, we have $\mathbf{H}_1 \mathbf{H}_1^H \!\sim\! \mathcal{W}_{M_e}\!\!\left(\!\frac{L_t(\!N_t\!-\!L_t\!)}{N_t},\mathbf{I}_{M_e}\!\right)$ and $\mathbf{H}_2 \mathbf{H}_2^H \sim \mathcal{W}_{M_e}\left(\frac{(N_t-L_t)^2}{N_t},\eta\mathbf{I}_{M_e}\right)$, where $\mathcal{W}_m(n,\mathbf{\Sigma})$ denotes an $m\times m$ Wishart matrix with $n$ degrees of freedom and $\mathbf{\Sigma}$ is the covariance matrix of each column.
Strictly speaking, the distribution of $\hat{\mathbf{H}}\hat{\mathbf{H}}^H$ is
complicated and intractable.
However, by applying the result in \cite{Secure_Zhu_1}, \cite{Nydick2012wishart} and according to the decomposition in \eqref{H2_2}, it follows that $\hat{\mathbf{H}}\hat{\mathbf{H}}^H$ can be accurately approximated as a Wishart matrix, i.e., $\hat{\mathbf{H}}\hat{\mathbf{H}}^H\sim \mathcal{W}_{M_e}(a,b\mathbf{I}_{M_e})$, where parameters $a$ and $b$ are chosen such that the first two moments of both sides in \eqref{H2_2} are identical, leading to
\begin{align}
ab=\frac{L_t(N_t-L_t)}{N_t}+\eta\frac{(N_t-L_t)^2}{N_t},
\label{ab1}
\end{align}
and
\begin{align}
ab^2= \frac{L_t(N_t-L_t)}{N_t}+ \eta^2\frac{(N_t-L_t)^2}{N_t}.
\label{ab2}
\end{align}
Solving \eqref{ab1} and \eqref{ab2}, we obtain $a$ and $b$ respectively in \eqref{a} and \eqref{b}, and get the derived result in \eqref{H2} by applying the propery in  \cite[Eq. (29)]{DAC1} of the Wishart matrix $\hat{\mathbf{H}}\hat{\mathbf{H}}^H\sim \mathcal{W}_{M_e}(a,b\mathbf{I}_{M_e})$ with large $a$.
\end{proof}

\section{Proof of Theorem~\ref{theorem_Rs}}
\label{proof_theorem_Rs}

We first recall the following lemma from \cite[Lemma 1]{Secure_Zhu_1}.
\begin{lemma}
\label{lemma_sec_rate}
The ergodic achievable secrecy rate of Bob is given by
\begin{align}
\label{secure rate}
R_S=[R_U-C_E]^+,
\end{align}
where $R_U$ is the ergodic achievable rate of Bob and $C_E$ is the ergodic capacity between Alice and Eve seeking to decode the information sent to Bob.
\end{lemma}

Then, using Lemmas~\ref{lemma_preliminary_1}-\ref{lemma_preliminary_2} in Appendix~\ref{Appendix_lemmas}, we prove Theorem~\ref{theorem_Rs} by deriving $R_U$ and  $C_E$ separately.

The correlation matrices of $\widetilde{\mathbf{x}}_V$ and $\mathbf{n}_{AN}$ are first given by
\begin{align}
\mathbf{C}_X=\mathbb{E}\left\{\widetilde{\mathbf{x}}_V\widetilde{\mathbf{x}}_V^{H}\right\}=\mathbb{E}\left\{\mathbf{Ws}\mathbf{s}^{H}\mathbf{W}^H\right\}=\frac{\phi P}{M_t}\mathbf{I}_{M_t},
\label{Cx}
\end{align}
and
\begin{align}
\mathbf{C}_{AN}=\mathbb{E}\left\{\mathbf{n}_{AN} \mathbf{n}_{AN}^H\right\}
=\frac{(1-\phi)P}{N_t-L_t}\mathbf{I}_{N_t-L_t},
\label{C_AN}
\end{align}
respectively, where we use the result in \eqref{P_AN} and $\mathbf{W}=\sqrt{\frac{\phi P}{M_t}}\mathbf{I}_{M_t}$.
Then, according to \eqref{yv_til} and by substituting \eqref{Cx} and \eqref{C_AN}, the ergodic achievable rate of Bob can be expressed as
\begin{align}
R_U 
=&\mathbb{E} \left\{\log_2 \left|\mathbf{I}_{M_r} +  \bar{\mathbf{G}}\mathbf{C}_X \bar{\mathbf{G}}^H  \left(\hat{\mathbf{G}} \mathbf{C}_{AN}\hat{\mathbf{G}}^H + \sigma_n^2\mathbf{I}_{M_r}  \right)^{-1} \right|\right\}
\\
=&\mathbb{E} \left\{\!\log_2\! \left|\!\mathbf{I}_{M_r} \!+\!  \frac{\phi P}{M_t\sigma_n^2} \bar{\mathbf{G}}\bar{\mathbf{G}}^H  \!\!\left[\! \frac{(\!1\!-\!\phi\!)P}{(\!N_t\!-\!L_t\!)\sigma_n^2} \hat{\mathbf{G}}\hat{\mathbf{G}}^H \!+\! \mathbf{I}_{M_r}  \!\right]^{\!-\!1} \!\right|\!\right\}\!
\\
=&\mathbb{E} \left\{\!\log_2 \left|\mathbf{I}_{M_r} \!+\!  \frac{\phi P}{M_t\sigma_n^2} \left(1-\frac{\mathcal{F}(\alpha,\beta)}{4\alpha\beta}\right) \bar{\mathbf{G}} \bar{\mathbf{G}}^H \right|\!\right\}
\label{Cu_2}\\
=&\mathbb{E} \left\{\!\log_2 \left|\mathbf{I}_{M_t} \!+\!  \frac{\phi P}{M_t\sigma_n^2} \left(1-\frac{\mathcal{F}(\alpha,\beta)}{4\alpha\beta}\right) \bar{\mathbf{G}}^H \bar{\mathbf{G}} \right|\!\right\}
\label{Cu_add}\\
=&\log_2 \left|\mathbf{I}_{M_t} +  \frac{M_r \phi P}{M_t\sigma_n^2} \left(1-\frac{\mathcal{F}(\alpha,\beta)}{4\alpha\beta}\right)\mathbf{I}_{M_t} \right|
\label{Cu_3},
\end{align}
where \eqref{Cu_2} applies Lemma~\ref{lemma_preliminary_1} in Appendix~\ref{Appendix_lemmas},
\eqref{Cu_add} applies the equality $|\mathbf{I}+\mathbf{AB}|=|\mathbf{I}+\mathbf{BA}|$,
and \eqref{Cu_3} comes from the fact that $\frac{1}{M_r} \bar{\mathbf{G}}^H \bar{\mathbf{G}}$ almost surely converges to $\mathbf{I}_{M_t}$ for large $M_r$.
Thus, $R_U$ in \eqref{Cu} is directly obtained from \eqref{Cu_3}.

On the other hand, according to \eqref{yev_til}, the ergodic capacity of Eve is expressed as
\begin{align}
C_E&=\mathbb{E} \left\{\log_2 \left|\mathbf{I}_{M_e} + \bar{\mathbf{H}}\mathbf{C}_X\bar{\mathbf{H}}^H \left(\hat{\mathbf{H}}\mathbf{C}_{AN}\hat{\mathbf{H}}^H \right)^{-1} \right|\right\}
\\
&=\mathbb{E} \left\{\log_2 \left|\mathbf{I}_{M_e} + \frac{\phi(N_t-L_t)}{(1-\phi)M_t} \bar{\mathbf{H}}\bar{\mathbf{H}}^H \left(\hat{\mathbf{H}}\hat{\mathbf{H}}^H \right)^{-1} \right|\right\}
\label{Ce_1}\\
&=\mathbb{E} \left\{\log_2 \left|\mathbf{I}_{M_e} + \frac{\phi(N_t-L_t)}{(1-\phi)M_t(a-M_e)b} \bar{\mathbf{H}}\bar{\mathbf{H}}^H  \right|\right\}
\label{Ce_2}\\
&=\mathbb{E} \left\{\log_2 \left|\mathbf{I}_{M_t} + \frac{\phi(N_t-L_t)}{(1-\phi)M_t(a-M_e)b} \bar{\mathbf{H}}^H\bar{\mathbf{H}}  \right|\right\}
\label{Ce_3}
,
\end{align}
where \eqref{Ce_1} uses \eqref{Cx} and \eqref{C_AN},
and \eqref{Ce_2} applies Lemma~\ref{lemma_preliminary_2} in Appendix~\ref{Appendix_lemmas}.
In \eqref{Ce_3}, $\frac{L_t M_t}{N_t}$ columns of $\bar{\mathbf{H}}$ follow the distribution $\mathcal{CN}(\mathbf{0},\mathbf{I}_{M_e})$ and the remaining $\frac{(N_t-L_t)M_t}{N_t}$ columns follow $\mathcal{CN}(\mathbf{0},\eta\mathbf{I}_{M_e})$.
Thus, $\bar{\mathbf{H}}^H\bar{\mathbf{H}}$ converges almost surely to a diagonal matrix with $\frac{L_t M_t}{N_t}$ diagonal entries equal to $M_e$ and the remaining $\frac{(N_t-L_t)M_t}{N_t}$ diagonal entries equal to $\eta M_e$.
This further yields the derived result in \eqref{Ce} from \eqref{Ce_3}.

Finally, the ergodic achievable secrecy rate in \eqref{Rs} is obtained by substituting \eqref{Cu} and \eqref{Ce} into \eqref{secure rate}.

\section{Proof of Theorem~\ref{theorem_Rs_bound}}
\label{proof_lemma_Rs}

Substituting \eqref{Cu_low}  and \eqref{Ce_app} into \eqref{Rs} of Theorem \ref{theorem_Rs}, the ergodic achievable secrecy rate at low SNR can be expressed as
\begin{align}
R_S^\mathrm{L}&\!=\!M_t\!\! \left[\!  \log_2\!\left(\! 1\!\!+\!\!\frac{M_r \phi P}{M_t \sigma_n^2}  \right)
\!-\!  \rho \log_2\!\left(\! 1\!\!+\!\frac{\phi M_e}{(\!1\!-\!\phi\!)M_t\!\left[\!\eta\!+\!(\!1\!-\!\eta\!)\rho\right]} \!\right)
\!\!-\!\! (\!1\!-\!\rho\!) \log_2\!\!\left( \!\!1\!+\!\frac{\phi M_e\eta}{(\!1-\phi\!)M_t\!\left[\!\eta\!+\!(\!1\!-\!\eta\!)\rho\right]}\!\right) \!\right]^{+}
\\
&\!\leq \!M_t \!\left[  \log_2\!\left(\! 1\!+\!\frac{M_r \phi P}{M_t \sigma_n^2} \! \right)
\!-\!  \rho \log_2 \frac{\phi M_e}{(\!1\!-\!\phi)M_t\left[\eta\!+\!(\!1\!-\!\eta)\rho\right]}
\!-\! (\!1\!-\!\rho) \log_2 \frac{\phi M_e\eta}{(1\!-\!\phi)M_t\!\left[\eta\!+\!(1\!-\!\eta)\rho\right]} \right]^+.
\label{Rs_low_2}
\end{align}
Considering a scenario with a powerful eavesdropper, i.e., $M_e\gg M_t$, an upper bound for $R_S^\mathrm{L}$ is given in \eqref{Rs_low_2}.
After some basic manipulations, the derived result in \eqref{Rs_low_bound} is obtained.

By substituting \eqref{Cu_high} and \eqref{Ce_app} into \eqref{Rs}, the ergodic achievable secrecy rate at high SNR can be expressed as
\begin{align}
R_S^\mathrm{H}\!=& M_t \!\left[\! \log_2\!\!\left(\! 1\!+\!\frac{M_r \phi }{M_t(\!1\!-\!\phi\!)\eta \left(\!1\!-\!\frac{M_r}{N_t(1\!-\!\rho)}\!\right)}  \!\right)
\!-\! \rho \log_2\!\left(\!\! 1\!+\!\frac{\phi M_e}{(\!1\!-\!\phi\!)M_t\left[\eta\!+\!(\!1\!-\!\eta\!)\rho\right]} \!\right)\right.
\nonumber\\&\left.
\!- (1\!-\!\rho) \log_2\!\left(\!\! 1\!\!+\!\!\frac{\phi M_e\eta}{(\!1\!-\!\phi\!)M_t\!\left[\eta\!+\!(\!1\!-\!\eta\!)\rho\right]} \!\right) \!\right]^+\!\!
\\
\leq& M_t \!\left[\! \log_2\!\!\left(1+\frac{M_r \phi }{M_t(\!1\!-\!\phi\!)\eta \left(\!1\!-\!\frac{M_r}{N_t(1-\rho)}\!\right)}  \!\right)
\!-\! \rho \log_2 \frac{\phi M_e}{(\!1\!-\!\phi\!)M_t\left[\eta\!+\!(\!1\!-\!\eta\!)\rho\right]} \right.
\nonumber\\&\left.
\!- (1\!-\!\rho) \log_2 \frac{\phi M_e\eta}{(\!1\!-\!\phi\!)M_t\!\left[\eta\!+\!(\!1\!-\!\eta\!)\rho\right]}   \!\right]^+\!\!
\label{Rs_high_2}\\
=&M_t \left[  \!\log_2\!\!\left(1+\frac{M_r \phi }{M_t(\!1\!-\!\phi\!)\eta \left(\!1\!-\!\frac{M_r}{N_t(1-\rho)}\!\right)}  \!\right)
\!-\!  \log_2 \frac{\phi M_e}{(\!1\!-\!\phi\!)M_t}
\!+\! \log_2\Big(\!\eta^{\rho\!-\!1}\left[\eta\!+\!(1\!-\!\eta)\rho\right] \Big)\! \right]^+.
\label{Rs_high_3}
\end{align}
Similar to the low SNR case, an upper bound for $R_S^\mathrm{H}$ is obtained in \eqref{Rs_high_2} and \eqref{Rs_high_3} by considering $M_e\gg M_t$.
Furthermore, by adopting the condition of $M_r\gg M_t$ to guarantee a positive secrecy rate, the upper bound in \eqref{Rs_high_bound} is obtained.

\section{Proof of Lemma~\ref{lemma_rho_star}}
\label{proof_lemma_rho_star}
We prove Lemma~\ref{lemma_rho_star} by deriving the derivatives of $\chi_\mathrm{L}$ and $\chi_\mathrm{H}$ w.r.t. $\rho$.
For $\chi_\mathrm{L}$ in \eqref{g1}, the first derivative w.r.t. $\rho$ is expressed as
\begin{align}
\frac{\partial \chi_\mathrm{L} }{\partial \rho}=\eta^{\rho-1}\Big(1-\eta+\ln \eta \left[\eta+(1-\eta)\rho\right]\Big).
\label{g_rho}
\end{align}
By forcing $\frac{\partial \chi_\mathrm{L} }{\partial \rho}=0$, we obtain $\rho^*$ in \eqref{rho_star}, which is a saddle point of $\chi_\mathrm{L}$. Then we show that this saddle point maximizes $\chi_\mathrm{L}$ by checking its second derivative for two separate cases. The second derivative is given as
\begin{align}
\frac{\partial^2 \chi_\mathrm{L} }{\partial \rho^2}=\eta^{\rho-1} \ln\eta \Big(2-2\eta+\ln\eta \left[\eta+(1-\eta)\rho\right]\Big).
\label{g_rho_2}
\end{align}

1)
For $\eta\in[0.2032,1)$, it can be easily verified that $\frac{\partial^2 \chi_\mathrm{L} }{\partial \rho^2}\leq0$ for $\rho\in\left[\frac{M_t}{N_t},1\right]$. Thus $\chi_\mathrm{L}$ is a convex function w.r.t. $\rho$ which means that the saddle point $\rho^*$ in \eqref{rho_star} achieves the optimum of $\chi_\mathrm{L}$.

2)
For $\eta\in(0,0.2032)$, it can be shown that $\frac{\partial \chi_\mathrm{L} }{\partial \rho}>0$ for $\rho\in\left[\frac{M_t}{N_t},\rho^*\right)$ and $\frac{\partial \chi_\mathrm{L} }{\partial \rho}<0$ for $\rho\in(\rho^*,1]$.
Hence, $\rho^*$ in \eqref{rho_star} also achieves the largest $\chi_\mathrm{L}$ in this case.


For $\chi_\mathrm{H}$ in \eqref{g2}, the derivative w.r.t. $\rho$ is
\begin{align}
\frac{\partial \chi_\mathrm{H} }{\partial \rho}=&\frac{\eta^{\rho\!-\!2}}{\left(1\!-\!\rho\!-\!\frac{M_r}{N_t}\right)^2}
\left[\!\ln(\!1\!-\!\rho\!)[\!\eta\!+\!(\!1\!-\!\eta)\rho]\left(\!1\!-\!\rho\!-\!\frac{M_r}{N_t}\!\right) \right.
\left.\!+\! (\!1\!-\!\eta\!)(\!1\!-\!\rho\!)^2\!+\!\frac{M_r}{N_t}(\!2\rho\!+\!2\eta\!-\!2\rho\eta\!-\!1\!)\! \right]
.
\label{g2_rho}
\end{align}
Generally, $\frac{\partial \chi_\mathrm{H} }{\partial \rho}>0$ for small $\rho$ while $\frac{\partial \chi_\mathrm{H} }{\partial \rho}<0$ for large $\rho$.
Using the fact $\frac{M_r}{N_t}\rightarrow 0$ for a large $N_t$ in \eqref{g2_rho}, we have
\begin{align}
\frac{\partial \chi_\mathrm{H} }{\partial \rho}&\rightarrow \eta^{\rho-2}\Big(1-\eta+\ln \eta \left[\eta+(1-\eta)\rho\right]\Big)
=\frac{1}{\eta}\frac{\partial \chi_\mathrm{L} }{\partial \rho}
.
\label{g2_rho_2}
\end{align}
According to $\chi_\mathrm{L}$, we thus have that $\rho^*$ in \eqref{rho_star} is also a saddle point that maximizes $\chi_\mathrm{H}$.

\end{appendices}

\end{document}